%% file: main.tex
\documentclass[lettersize,journal]{IEEEtran}
\usepackage{amsmath,amssymb,amsfonts,leftidx}
\usepackage{algorithm}
\usepackage{algpseudocode}
\usepackage{array}
\usepackage[caption=false,font=normalsize,labelfont=sf,textfont=sf]{subfig}
\usepackage{textcomp}
\usepackage{stfloats}
\usepackage{url}
\usepackage{verbatim}
\usepackage{graphicx}
\usepackage{cite}
\usepackage{colortbl}
\usepackage[utf8]{inputenc}
\usepackage{siunitx}
\DeclareSIUnit{\belisotropic}{Bi}
\DeclareSIUnit{\dBi}{\deci\belisotropic}
\usepackage{graphicx,color}
\usepackage{comment}
\usepackage{tikz}
\usetikzlibrary{shapes.geometric, arrows}
\usetikzlibrary{arrows.meta}
\usetikzlibrary{decorations.pathreplacing,calligraphy}
\usepackage{xcolor}
\usepackage{lipsum}
\usepackage[nolist]{acronym}
\usepackage{pgfplots}
\usepackage{mathtools}
\usepackage{bm}
\usepackage{amsthm}
\usepackage[dvipsnames]{xcolor}
\DeclareMathOperator*{\argmax}{arg\,max}
\DeclareMathOperator*{\argmin}{arg\,min}

\def\BibTeX{{\rm B\kern-.05em{\sc i\kern-.025em b}\kern-.08em
    T\kern-.1667em\lower.7ex\hbox{E}\kern-.125emX}}
\AtBeginDocument{\definecolor{ojcolor}{cmyk}{0.93,0.59,0.15,0.02}}
\usepackage{balance}
\pgfplotsset{compat=1.18}
\begin{document}

\input{acronyms.tex}

\include{commands.tex}

\title{On the Outage Probability of Multiuser  Multiple Antenna Systems with Non-Orthogonal Multiple Access  for Air-Ground Communications}

\author{Ayten Gürbüz and Giuseppe Caire \IEEEmembership{(Fellow, IEEE)}

\thanks{Ayten Gürbüz is with the Institute of Communications and Navigation, German Aerospace Center (DLR), Oberpfaffenhofen, 82234 Wessling, Germany (e-mail: Ayten.Guerbuez@dlr.de).}
\thanks{Giuseppe Caire is with  the Faculty of Electrical Engineering and Computer Science, Technical University of Berlin, 10587 Berlin, Germany (e-mail: caire@tu-berlin.de).}

}

\IEEEoverridecommandlockouts
\IEEEpubid{\begin{minipage}{\textwidth}\ \\ \\ \\ \\ [12pt]
  \large This work has been submitted to the IEEE for possible publication. Copyright may be transferred without notice, after which this version may no longer be accessible.
\end{minipage}}


\maketitle

\begin{abstract}
This paper explores multiuser multiple antenna systems as a means to enhance the spectral efficiency of aeronautical communications systems.  To this end, the outage regime for a multiuser multiple antenna system is studied within a realistic geometry-based stochastic \ac{ag} channel model.
In this application, users (aircraft) transmit air traffic management data to the ground station at a predefined target rate.
Due to the nature of the \ac{ag} propagation, we argue that the relevant performance metric in this context is the information outage probability.
We consider the outage probability \rewthree{of individual aircraft} under three decoding approaches. The first is based on  \ac{sic}. The second extends the first approach by considering joint group decoding. The third is a version of the second that limits the size of the jointly decoded user groups in order to lower the decoding complexity.
The results show that joint group decoding, even in groups of only two, can significantly increase the spectral efficiency in the AG channel by allowing a large number of aircraft to transmit over a non-orthogonal channel with very low outage probabilities.
\end{abstract}

\input{Sections/01-intro}

\input{Sections/02}

\input{Sections/03}
\input{Sections/04}

\input{Sections/05}

\input{Sections/06}
\input{Sections/07}
\input{Sections/08}

\input{Sections/Acknowledgement}
\input{Sections/Appendix}

\bibliographystyle{IEEEtran}
\bibliography{IEEEabrv,ref}
\newpage
\vspace{-85pt}
\begin{IEEEbiography}[{\includegraphics[width=1in,height=1.25in,clip,keepaspectratio]{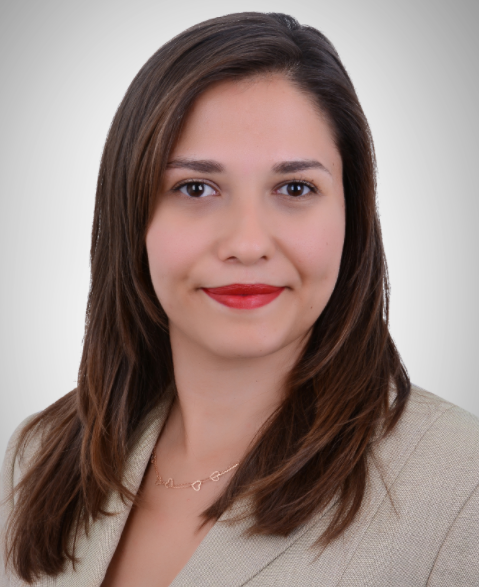}}]{Ayten Gürbüz}
was born in 1995. She received the B.Sc. in electrical and electronics engineering from Bilkent University in Ankara (Turkey) in 2018. She subsequesntly completed the M.Sc. in communications engineering program at the Technical University of Munich in January 2021.  She is currently a researcher at the German Aerospace Center (DLR) and a Ph.D. candidate at the Technical University of Berlin. Her research interests include the application of multiple antenna systems to aeronautical communications, with a particular interest in multiuser scenarios and \mbox{non-orthogonal multiple access (NOMA).}
\end{IEEEbiography}
\vspace{-90pt}
\begin{IEEEbiography}[{\includegraphics[width=1in,height=1.25in,clip,keepaspectratio]{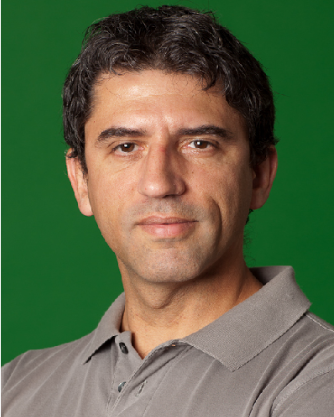}}]{Giuseppe Caire}
(Fellow, IEEE) was born in Turin in 1965. He received the B.Sc. degree in electrical engineering from Politecnico di Torino in 1990, the M.Sc. degree in electrical engineering from Princeton University in 1992, and the Ph.D. degree from Politecnico di Torino in 1994. 

He was a Post-Doctoral Research Fellow with European Space Agency (ESTEC), Noordwijk, The Netherlands, from 1994 to 1995; an Assistant Professor of telecommunications with Politecnico di Torino; an Associate Professor with the University of Parma, Italy; a Professor with the Department of Mobile Communications, Eurecom Institute, Sophia-Antipolis, France; and a Professor of electrical engineering with the Viterbi School of Engineering, University of Southern California, Los Angeles. He is currently an Alexander von Humboldt Professor with the Faculty of Electrical Engineering and Computer Science, Technical University of Berlin, Germany. His research interests include communications theory, information theory, and channel and source coding, with a particular focus on wireless communications. He received the Jack Neubauer Best System Paper Award from the IEEE Vehicular Technology Society in 2003, the IEEE Communications Society and Information Theory Society Joint Paper Award in 2004 and 2011, the Okawa Research Award in 2006, the Alexander von Humboldt Professorship in 2014, the Vodafone Innovation Prize in 2015, the ERC Advanced Grant in 2018, the Leonard G. Abraham Prize for the Best IEEE JOURNAL ON SELECTED AREAS IN COMMUNICATIONS Paper in 2019, the IEEE Communications Society Edwin Howard Armstrong Achievement Award in 2020, the 2021 Leibniz Prize of German National Science Foundation (DFG), and the CTTC Technical Achievement Award of the IEEE Communications Society in 2023. He has served on the Board of Governors of the IEEE Information Theory Society from 2004 to 2007, where he was an Officer from 2008 to 2013. He was the President of the IEEE Information Theory Society in 2011.
\end{IEEEbiography}

\end{document}

%% file: acronyms.tex
\begin{acronym}
    \acro{faa}[FAA]{Federal Aviation Administration}
    \acro{euro}[EUROCONTROL]{European Organization for the Safety of Air Navigation}
    \acro{icao}[ICAO]{International Civil Aviation Organisation}
    \acro{fci}[FCI]{future communications infrastructure}
    \acro{sesar}[SESAR]{Single European Sky ATM Research}
    
    \acro{atm}[ATM]{air traffic management}
    \acro{dme}[DME]{distance measuring equipment}
    \acro{ldacs}[LDACS]{L-band Digital Aeronautical Communications System}
    \acro{vhf}[VHF]{very high-frequency}

    \acro{as}[AS]{air station}
    \acro{gs}[GS]{ground station}
    \acro{ag}[AG]{air-ground}
    \acro{rl}[RL]{reverse link}
    \acro{fl}[FL]{forward link}
    
    \acro{msl}[MSL]{mean sea level}
    \acro{nm}[NM]{nautical miles}
    \acro{rf}[RF]{radio frequency}
    
    \acro{tl}[TL]{takeoff \& landing}
    \acro{cd}[CD]{climb \& descent}
    \acro{ec}[EC]{enroute cruise}
    
    \acro{fspl}[FSPL]{free-space path loss}
    \acro{mpc}[MPC]{multipath component}
    \acro{los}[LOS]{line-of sight}
    \acro{gmp}[GMP]{ground multipath component}
    \acro{sr}[SR]{specular reflection}
    \acro{lmp}[LMPs]{lateral multipath components}
    
    \acro{upa}[UPA]{uniform planar rectangular array}
    \acro{uca}[UCA]{uniform circular array}
    \acro{sa}[SA]{single-omnidirectional antenna}
    \acro{ma}[MA]{multiple antenna}
    \acro{hpbw}[HPBW]{half-power beamwidth}
    \acro{siso}[SISO]{single input single output}
    \acro{simo}[SIMO]{single input multiple output}
    \acro{miso}[MISO]{multiple input single output}
    \acro{mimo}[MIMO]{multiple input multiple output}
    \acro{su-simo}[SU-SIMO]{single-user single input multiple output}
    \acro{mu-mimo}[MU-MIMO]{multiuser multiple input multiple output}
    \acro{mu-simo mac}[MU-SIMO MAC]{multiuser single input multiple output multiple-access channel}
    \acro{mac}[MAC]{multiple-access channel}
    \acro{noma}[NOMA]{non-orthogonal multiple access}
    \acro{oma}[OMA]{orthogonal multiple access}
    \acro{cgtr}[CGTR]{channel gain and transmission rate}
    
    \acro{zf}[ZF]{zero forcing}
    \acro{ssa}[SSA]{single successive algorithm}
    \acro{gsa}[GSA]{group successive algorithm}
    \acro{lgsa}[LGSA]{limited group successive algorithm}
    \acro{vblast}[V-BLAST]{vertical-bell laboratories layered space-time}
    \acro{sic}[SIC]{successive interference cancellation}
    \acro{isu}[ISU]{independent single-user}
    \acro{mmse}[MMSE]{minimum mean squared error}
    \acro{mmse-sic}[MMSE-SIC]{minimum mean squared error successive interference canceller}
    \acro{zf-sic}[ZF-SIC]{zero forcing successive interference canceller}
    \acro{bps}[bps/Hz]{bits per second per hertz}
    
    \acro{snr}[SNR]{signal-to-noise ratio}
    \acro{sir}[SIR]{signal-to-interference-ratio}
    \acro{sinr}[SINR]{signal-to-noise-interference-ratio}
    \acro{ber}[BER]{bit error rate}
    \acro{awgn}[AWGN]{additive white Gaussian noise}
    \acro{qpsk}[QPSK]{quadrature phase shift keying}
    \acro{cr}[CR]{coding rate}
    \acro{rs}[RS]{receiver sensitivity}

    \acro{pdf}[PDF]{probability density function}
    \acro{cdf}[CDF]{cumulative density function}
    
    \acro{doa}[DOA]{direction of arrival}
    \acro{dod}[DOD]{direction of departure}
    
    \acro{cir}[UCA]{uniform circular array}
    \acro{rec}[UPRA]{uniform planar rectangular array}
    \acro{sun}[UFSA]{uniform Fermat spiral array}
  
    \acro{tdma}[TDMA]{time division multiple access}
    \acro{ofdma}[OFDMA]{orthogonal frequency-division multiple access}
    
    \acro{rhs}[RHS]{right hand side}
\end{acronym}

%% file: commands.tex
\newcommand{\ag}[1]{\textcolor{red}{#1}}
\newcommand{\gc}[1]{\textcolor{orange}{#1}}
\newcommand{\addrf}[1]{\textcolor{red}{#1}}

\newcommand{\rewone}[1]{\textcolor{black}{#1}}
\newcommand{\rewtwo}[1]{\textcolor{black}{#1}}
\newcommand{\rewthree}[1]{\textcolor{black}{#1}}
\newcommand{\rewfour}[1]{\textcolor{black}{#1}}

\setlength{\abovecaptionskip}{0pt}
\setlength{\belowcaptionskip}{3pt}
\setlength{\intextsep}{1mm}
\setlength{\textfloatsep}{1mm}
\setlength{\dblfloatsep}{2mm}
\setlength{\dbltextfloatsep}{2mm}




\newlength\myindent
\setlength\myindent{2em}
\newcommand\bindent{%
  \begingroup
  \setlength{\itemindent}{\myindent}
  \addtolength{\algorithmicindent}{\myindent}
}
\newcommand\eindent{\endgroup}

\newtheorem{theorem}{Theorem}[section]
\newtheorem{corollary}{Corollary}[theorem]
\newtheorem{lemma}[theorem]{Lemma}

\algnewcommand{\algorithmicforeach}{\textbf{for each}}
\algdef{SE}[FOR]{ForEach}{EndForEach}[1]
  {\algorithmicforeach\ #1\ \algorithmicdo}
  {\algorithmicend\ \algorithmicforeach}
\algnewcommand{\LineComment}[1]{\State \(\triangleright\) #1}

\algrenewcommand\algorithmicrequire{\textbf{Input:}}
\algrenewcommand\algorithmicensure{\textbf{Output:}}

\newcommand\Algphase[1]{%
\vspace*{-.4\baselineskip}\Statex\hspace*{\dimexpr-\algorithmicindent-2pt\relax}\rule{0.49\textwidth}{0.1pt}%
\Statex\vspace{-.3\baselineskip}\hspace*{-\algorithmicindent}{#1}%
\vspace*{-.7\baselineskip}\Statex\hspace*{\dimexpr-\algorithmicindent-2pt\relax}\rule{0.49\textwidth}{0.1pt}%
}

\newcommand\AlgphaseV[1]{%
\Statex\hspace*{-\algorithmicindent}{#1}%
\vspace*{-.7\baselineskip}\Statex\hspace*{\dimexpr-\algorithmicindent-2pt\relax}\rule{0.49\textwidth}{0.4pt}%
}




\algnewcommand\algorithmicIfT{\textbf{if}}
\algdef{SE}[IFT]{IfT}{EndIfT}[1]
  {\algorithmicIfT\ #1}
  {\algorithmicend\ \algorithmicIfT}

\algtext*{EndIf}
\algtext*{EndFor}
\algtext*{EndIfT}
\algtext*{EndWhile}
\algtext*{EndForEach}
\algtext*{EndFunction}

%% file: Sections/01-intro.tex
\section{Introduction}\label{sec:intro}
\acresetall 
\IEEEPARstart{A}{ll} 
aircraft operating within a controlled airspace are managed by a globally standardized \ac{atm} system to ensure the efficient traffic flow and flight safety. Air traffic controllers on the ground monitor the air traffic and give commands to the pilots on board the aircraft. In addition, they accept/reject requests from pilots to change their flight path. However, the demand for air transportation is continuously growing and the current systems are  expected to reach their capacity in some world regions within the coming years \cite{ challenge2}. 
Current aeronautical communications systems operate in the \ac{vhf} band. Digital data links operating in the  \ac{vhf}  band have been developed in the 1990s and only provide a limited data rate, well below the data rate required for a full \ac{atm} modernization in the future \cite{cocr}.

A key limitation in aeronautical communication technology is the spectrum to be used. Communications supporting \ac{atm} exchanges are classified as ``safety of life" communications, which grants them special protection to prevent interference. However, this protection applies only to certain frequency bands in the VHF, L-band, and C-band \cite{Fistas11}. 
Considering the congestion and propagation characteristics of these bands, the \ac{icao} recommended the use of the frequency band between 960-1164~MHz in the L-band for the future aeronautical communications on a secondary basis \cite{icao05,wrc-07}.
Operating on a secondary basis means that systems deployed in this band must not interfere with already existing legacy systems in the L-band. This imposes limits on the transmit power and spectrum usage. Therefore, the system must optimize the use of the limited spectrum while meeting the associated power constraints.
To modernize the \ac{atm} system, two major research frameworks have been launched: the \ac{sesar} \cite{sesar} in Europe and the Next Generation National Airspace System \cite{nextgen} in the United States. Both researches  are being conducted under the framework of \ac{icao}.

The growing demands of aeronautical communications can be addressed by multiple antenna systems, which are known to improve spectral efficiency.  In \cite{ayten}, we presented achievable upper boundaries for a single-user multiple antenna system and demonstrated that the performance is limited by the propagation characteristics of the \ac{ag} communications. An alternative to  point-to-point multiple antenna systems is a multiuser system, \cite{muMimo1,muMimo2,muMimo3,emil}, where an antenna array receives signals from multiple users simultaneously. In this case, the users are not separated in time or frequency; this is known as \ac{noma}. Multiuser systems are more resilient to the propagation environment, allowing them to exploit the full potential of multiple antenna technology. 

Multiuser multiple antenna systems have been extensively studied in cellular communications 
\cite{noma_survey, vblast, order2, vehicle4}.
In addition, numerous studies investigate the use of unmanned aerial vehicles and/or satellites to support ground networks in cellular communications, e.g., \cite{vehicle1, vehicle3, vehicle5, vehicle2, outage2}. 
Despite these advances, the research on multiuser multiple antenna systems with \ac{noma} for \ac{atm} systems remains limited. 
To the best of our knowledge, the few available works include \cite{air1} and our recent study \cite{ayten_icc25}.
\rewfour{However, the study in \cite{air1} does not thoroughly consider the realistic propagation characteristics of the \ac{ag} channel and focuses solely on the achievable ergodic sum rate \cite{muMimo4}. This is in line with the common trend in the \ac{noma} literature, where the primary performance objective is the maximization of the system sum rate or weighted sum rate \cite{ noma_survey, vehicle1, vehicle3, vblast, order2, vehicle4, vehicle5}.
Such metrics are well-suited for applications such as video streaming and multimedia delivery that require transmission of large data packages. However, these performance metrics do not guarantee reliable communication, which is essential in ``safety of life" applications such as the one considered in this work. In particular, while temporary service degradations are generally tolerable in many communication scenarios discussed in the literature, they are strictly unacceptable in \ac{atm} communications. In this context, short, flight-critical messages must be exchanged reliably and with low latency.
Motivated by this requirement, we assume that each aircraft must be able to transmit at a target data rate, referred to as the guaranteed rate, ensuring that flight-critical information can be exchanged within acceptable delay bounds. Based on this assumption, we focus on minimizing the outage probability for individual aircraft at the guaranteed data rate 
under varying channel conditions.
}

\rewfour{
Outage probability has also been studied as a performance metric in numerous works, albeit less extensively than sum rate maximization. However, most existing studies characterize outage performance as a function of the \ac{snr} \cite{noma_survey, vehicle2, outage2}. 
In contrast, we consider a fixed transmission power that complies with the strict regulations of the protected portion of the L-band and evaluate the outage probability at the system's guaranteed rate, while the \ac{snr} varies across aircraft depending on their distance to the \ac{gs} and the effects of multipath fading.}
\rewfour{In addition, prior work on satellite communications \cite{vehicle2, outage2} highlights that outage probability is strongly influenced by the propagation environment. Since drones and satellites differ substantially from typical aircraft in terms of speed, altitude, and communication range, they exhibit distinct propagation characteristics, and the corresponding results cannot be directly applied to typical aircraft. Motivated by this observation, both our recent work \cite{ayten_icc25} and the present paper incorporate realistic \ac{ag} propagation characteristics. 
}

\rewthree{
The \ac{ag} channel has been extensively characterized through measurement campaigns \cite{nik_mea, matolak, sun,matolak3}, providing a solid foundation for accurate channel modeling. }
We adapt the geometry-based stochastic \ac{ag} channel model from \cite{nik_tvt1}, which was developed based on the measurement campaign in \cite{nik_mea}. Originally designed for a single antenna system, we extend it to a multiple antenna system with \ac{noma} setup using the insights from \cite{ayten2} and we derive the corresponding channel matrix.
\rewthree{While our recent study \cite{ayten_icc25} employed this realistic channel model to analyze channel estimation under high Doppler shifts across various flight scenarios, in this study we assume perfect channel knowledge at the receiver in order to investigate the theoretical limits in the outage regime.}

Due to the strict nature of the \ac{atm} systems, the constraints considered in this paper are as follows: the aircraft's transmission power is fixed according to the \ac{sesar} project specifications \cite{sesar,LDACSSpec19}, aligning with the power limits of the protected L-band. This fixed transmission power is necessary to maintain simplicity in the onboard system. For the same reason, we also assume a single antenna element on the aircraft. On the other hand, we consider an antenna array at the \ac{gs}. All aircraft in the system must meet the guaranteed rate to ensure the reliable transmission of flight critical messages. We consider two cases: a) an equal-rate system in which all aircraft transmit at the system's guaranteed rate, and 
b) a variable-rate system, where aircraft transmit at least at the guaranteed rate.

Using the extended channel model and considering the constraints and requirements of aeronautical communications systems, this paper focuses on the fundamental research of multiuser multiple antenna technology  for \ac{ag} communications. 
The main contributions are listed as follows:
\begin{itemize}
    \item We study the  multiuser multiple antenna system technology with \ac{noma} within a realistic geometry-based stochastic \ac{ag} channel model, which is inherently quasi-static.  \rewthree{Using this quasi-static \ac{noma} channel model, we examine the information outage probability under partial decoding, where some aircraft can be successfully decoded while others experience outage.}
    \item  We evaluate outage probability \rewthree{of individual aircraft} under three decoding algorithms. These algorithms minimize the outage probability for either only \ac{sic} decoders or  \ac{sic} decoders in conjunction with joint group decoding. The proposed algorithms also derive practical strategies for implementing these decoding strategies.
    \item  The first algorithm is the \ac{ssa}, which minimizes the outage probability for \ac{sic} decoders. We compare \ac{ssa} with the decoding order strategies proposed in \cite{vblast} and \cite{order2}. We show that \ac{ssa} outperforms the decoding ordering strategy in \cite{order2}. Moreover, we prove that the \ac{vblast} algorithm \cite{vblast} finds an optimal solution that minimizes the outage probability in equal-rate systems but it cannot obtain an optimal decoding order in variable-rate systems. Hence, \ac{ssa} outperforms V-BLAST in variable-rate systems.  
    \item The second algorithm is the \ac{gsa}. It determines the minimum achievable outage probability for a given capacity region. To achieve this minimum outage probability, it is necessary to use the \ac{sic} decoders together with joint group decoders. \ac{gsa} significantly reduces the outage probability compared to \ac{ssa}. The difference between \ac{gsa} and \ac{ssa}  becomes more apparent when the number of transmitting aircraft is large.
    \item The third algorithm is the \ac{lgsa}. It is a version of \ac{gsa} that forces the joint group decoding to be done in small groups. It does not obtain the optimal solution like \ac{gsa}, but it provides a more feasible solution for real-world applications while still proposing a significant improvement over \ac{ssa}.
\end{itemize}

Subsequently, the paper is organized as follows: in Section~\ref{sec:system_model}, we introduce the fixed parameters of the system. In Section~\ref{sec:channel}, we first provide background on the \ac{ag} channel propagation,  and follow this with an introduction to  the channel model used in the simulations. Moving to Section~\ref{sec:capacity}, we explore the interplay between the capacity region and outage probability 
and demonstrate that the calculation of the outage probability is a challenging problem. 
In Section~\ref{sec:algs}, we present algorithms designed to solve the challenges outlined in Section~\ref{sec:capacity}. Subsequently, in Section~\ref{sec:receiver}, we present two existing decoding ordering algorithms in the literature. Proceeding to Section~\ref{sec:results},   we evaluate the performance and computational complexity of the proposed algorithms and compare their performance with the existing decoding order algorithms in the literature. Finally, in Section~\ref{sec:discussion} we summarize the obtained results and give an outlook for the future.

%% file: Sections/02.tex
\section{System Model}
\label{sec:system_model}

\begin{figure*}[!t]
     \centering
     \includegraphics[width=0.80\textwidth]{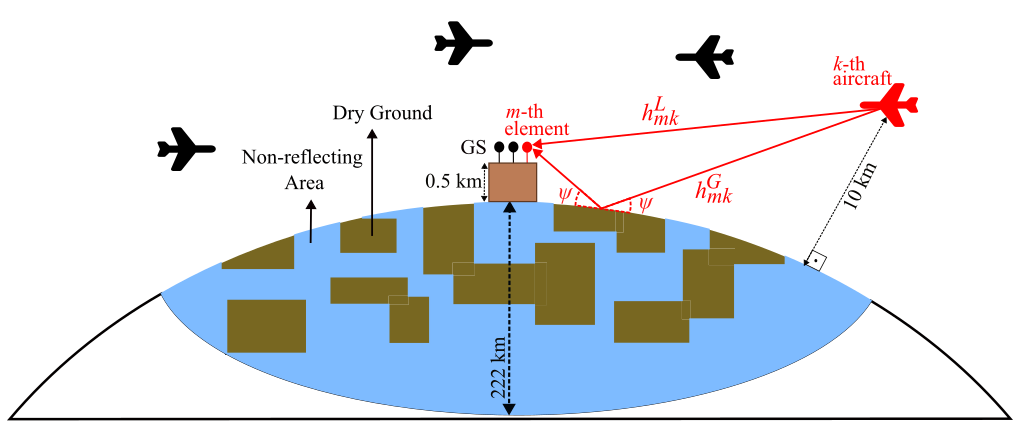}
     \caption{Simulation setup illustrating the system and channel model.}
    \label{fig:system}
\end{figure*}

In this work, we study a communications system between a \ac{gs} and multiple aircraft.  The \ac{gs} has an antenna array consisting of $M$ antenna elements,  while each of the $K$ aircraft is equipped with a single antenna. 
In order to meet the constraints of the protected L-band spectrum for \ac{atm} communications, the system model parameters are chosen based on the specifications in \cite{LDACSSpec19}. 
We assume that all $K$ aircraft transmit simultaneously to the \ac{gs}, at the carrier frequency of \qty{987}{\MHz} \cite{LDACSSpec19}. Thus, we consider a  multiuser multiple antenna system with \ac{noma} in the reverse link. 
To comply with the transmit power constraints, the transmit power of each aircraft is set to $P = \qty{41}{dBm}$, and the noise power at each receiver antenna is assumed to be $N_0 = \qty{-107}{dBm}$ \cite{LDACSSpec19}.

Furthermore, in this work a curved Earth model with a radius of \qty{6371}{km} is used \cite{Geodetic}. The \ac{gs} is assumed to be at the center of a circular cell with a radius of \qty{222}{km} (120 nautical miles) \cite{LDACSSpec19} and positioned \qty{500}{m} above the \ac{msl}, as tested in \cite{LDACS_journal,nik_mea}.
In this paper, we focus only on the \ac{ec} scenario. 
To this end, the altitude of the $K$ aircraft is set at \qty{10}{km} above \ac{msl}, a typical altitude for commercial flights, which corresponds to above flight level 300 \cite{separation2}.
Based on Instrument Flight Rules \cite{separation2}, we maintain a minimum separation of \qty{10}{km} (approximately 5 nautical miles) between any two aircraft to ensure realistic spacing and avoid overlap. We repeatedly simulate the positions of the aircraft within the given constraints to calculate the outage probability.

The evaluations are performed for a \ac{rec} where the antenna elements are arranged in a $\sqrt{M} \times \sqrt{M}$ grid, with each element evenly spaced at half-wavelength intervals.
We assume that the antenna elements have an isotropic radiation pattern and that the mutual coupling between the array elements is neglected, which is a common assumption in the literature on multiuser multiple antenna systems \cite{muMimo1,muMimo2,muMimo3,emil}.

%% file: Sections/03.tex
\section{Air-Ground Propagation and Channel Model}\label{sec:channel}
An understanding of the propagation effects is crucial for accurately modeling a realistic channel. In this section, we first provide a summary of the air-ground propagation insights based on the measurement campaigns conducted in the L-band \ac{ag} channel \cite{nik_mea,matolak,sun,matolak3}. Following this, we describe how we utilize the \ac{ag} channel model presented in \cite{nik_tvt1}.

\subsection{Air-Ground Propagation} \label{subsec:propagation}

According to the measurement campaigns \cite{nik_mea,matolak,sun,matolak3} the most dominant \ac{mpc} in the L-band \ac{ag} channel is the \ac{los} path, which is the direct path between the aircraft and the \ac{gs}. 
The \ac{los} path has significantly higher power than the other \ac{mpc}s and that it is almost always present.

The second strongest \ac{mpc} is the ground reflection, also called the \ac{gmp}. This is a \ac{sr} where the reflecting point is located on the direct line between the projection of \ac{gs} and the projection of the aircraft on the ground. The \ac{gmp} reaches the receiver shortly after the \ac{los} signal and with a very similar Doppler frequency compared to the \ac{los} signal. Therefore, it interferes constructively or destructively with the \ac{los}. In \cite{nik_mea}, \ac{los} attenuations of up to \qty{20}{dB} and amplifications of up to \qty{6}{dB} have been reported as a result of the \ac{gmp}.
The presence and strength of the \ac{gmp} depend on terrain characteristics, such as whether the aircraft is over water, urban areas, or hilly–mountainous regions \cite{matolak,sun,matolak3}.


The \ac{mpc}s that reflect off buildings, large structures, or vegetation are known as \ac{lmp}. These reflectors are usually located near the \ac{gs}. The \ac{lmp} typically have much lower power than both \ac{los} and \ac{gmp}.


Fading channels, which result from the constructive and destructive superposition of the various propagation paths, are known to have fluctuations. 
The duration of each fluctuation,  also known as the ``coherence time", is approximately inversely proportional to the Doppler spread.
The measurement results in \cite{nik_mea} indicate that the Doppler spread of  the \ac{ag} channel differs in three main phases of a typical flight scenario, which are \ac{tl}, \ac{cd}, and \ac{ec}.
In the \ac{ec} phase, all propagation paths experience the same Doppler shift as the \ac{los} path due to large distance between the transmitter and receiver and the resulting geometry. This leads to a very small Doppler spread.  
In the \ac{cd}, the Doppler spread is slightly larger than in the \ac{ec} phase. According to \cite{nik_phd}, which is also based on the measurement campaign reported in \cite{nik_mea}, the Doppler spread in this phase is \qty{60}{Hz}. 
In the \ac{tl} phase the Doppler spread is much higher, resulting in shorter coherence time.

In this paper, we focus only on the \ac{ec} scenario. Although \cite{nik_mea} do not report the exact Doppler spread of \ac{ec} scenario, it is known to be lower than in the \ac{cd} phase (\qty{60}{Hz} \cite{nik_mea}), which results in a coherence time of at least $1/\qty{60}{Hz}=\qty{16}{ms}$. Taking the example of the \ac{sesar} system specifications \cite{LDACSSpec19}, the duration of each codeword for aircraft to \ac{gs} transmission varies between \qty{0.72}{ms} to \qty{7.2}{ms}. Since these codeword durations are well below the \ac{ec} coherence time, it is reasonable to consider the channel as random but constant over a whole codeword, in accordance with the ``capacity versus outage approach" introduced in \cite{ozarow} for the so-called ``block-fading channel".

\subsection{Channel Model}
In this work, we utilize a simplified version of the proposed geometry-based stochastic channel model for the \ac{ag} channel in \cite{nik_tvt1}. The channel model parameterization is based on an L-band airborne measurement campaign  that took place at a regional airport in 2013, which is described in detail in \cite{nik_mea}.

\rewone{First, we generate a geometric scattering environment based on the statistical distributions in \cite{nik_tvt1}. While the scattering environment remains fixed throughout the simulations, the positions of the aircraft are randomly and repeatedly changed within this environment.}
We consider only the two strongest propagation effects in \cite{nik_tvt1}: the \ac{los} and the \ac{gmp}. 
We assume that the \ac{los} is always present and that it is modelled as the shortest possible path between the aircraft and the \ac{gs}. The presence of the \ac{gmp}, however, depends on the position of the aircraft \rewone{within the scattering environment}.

\rewone{The scattering environment is generated by } 
randomly characterizing the reflecting and non-reflecting areas on the ground. Figure~\ref{fig:system} shows an example ground reflecting area realization. In the simulations, the ratio between the total area of all reflecting surfaces and the overall area around the \ac{gs} is 50\%, as it is approximated in \cite{nik_tvt1} for a regional airport environment. A \ac{gmp} is present when the ground reflection point falls within a reflecting area. The location of the ground reflector is determined by calculating the \ac{sr} point on the ground surface, which depends on the coordinates of both \rewone{the transmitter and receiver antennas~\cite{Parsons}.}
\rewone{The \ac{gmp} component is computed separately for each transmitter–receiver antenna pair. }
The reflective areas are modeled as rectangles, with their size always much larger than the wavelength of the carrier signal. 
In \cite{nik_tvt1}, the reflective surfaces are modeled as concrete, dry ground, and medium ground, based on the measurement campaign in \cite{nik_mea}. In our study, for simplicity, we model the reflective surfaces as dry ground only. 

As justified in \mbox{Section \ref{subsec:propagation}}, 
the channel can be treated as a random variable rather than a process  for each codeword transmission~\cite{ozarow}.
To this end, we compute the instantaneous realization of the channel based on the positions of the $K$ aircraft. 
The \ac{los}, ${h}^{L}_{mk}$, and the \ac{gmp}, ${h}^{G}_{mk}$, channels between the $m$-th antenna element and the $k$-th aircraft are each computed as
\begin{equation}
    {h}^{L}_{mk}=\alpha^{L}_{mk}\cdot \exp{ \left(-j2\pi\frac{d^{L}_{mk}}{\lambda} \right)} \text{ ,}
\end{equation}
\begin{equation}
    {h}^{G}_{mk}=\rho_{v}\cdot\alpha^{G}_{mk}\cdot\exp{\left(-j2\pi\frac{d^{G}_{mk}}{\lambda}\right)}\text{ .}
\end{equation}
The path strengths for these channels are denoted by $\alpha_{mk}^{L}$ and $\alpha_{mk}^{G}$, respectively. These values are calculated using the \ac{fspl} formula. The signal wavelength is denoted by $\lambda$.
The propagation path lengths for these channels are given by $d_{mk}^{L}$ and $d_{mk}^{G}$. \rewone{Since each antenna element in an array occupies a distinct physical position, $d_{mk}^{L}$ and $d_{mk}^{G}$ differ across the array, leading to distinct phase rotations at each receiving antenna element. This naturally embeds the angle-of-arrival  information in the model. }
The vertical reflection coefficient for the \ac{gmp} is denoted by $\rho_{v}$, which equals 0 when the ground reflection point lies within a non-reflective area. 
If the ground reflection point is in a reflective area, $\rho_{v}$ is calculated using the equation from \cite{Parsons}, which considers the grazing angle and the electromagnetic properties of dry ground as specified in~\cite{ITU}.

In the end, the channel between the $m$-th antenna element and  the $k$-th aircraft, ${h}_{mk}$, is calculated by
\begin{equation}
    {h}_{mk}={h}^{L}_{mk}+{h}^{G}_{mk}\text{.}
\end{equation}
The vector representing the channel between the antenna array on \ac{gs} and the $k$-th aircraft is then given by \mbox{$\textbf{h}_{:k}=[{h}_{1k}, {h}_{2k},\dots,{h}_{Mk}]^T\in \mathbb{C}^{M}$}, and the channel matrix between the $K$ aircraft and the \ac{gs} is \mbox{$\textbf{H}=[\textbf{h}_{:1}, \textbf{h}_{:2}, \dots, \textbf{h}_{:K}] \in \mathbb{C}^{M\times K}$}.

%% file: Sections/04.tex
\section{Capacity Region and Outage Probability}\label{sec:capacity}

The signal received  at \ac{gs} for \textit{one channel use} of the discrete-time baseband complex channel model is denoted as  $\textbf{y}_j\in \mathbb{C}^{M}$ and it is
 related to the channel matrix $\textbf{H}\in \mathbb{C}^{M\times K}$ by
\begin{equation}
   \textbf{y}_j=\textbf{H}\textbf{x}_j+\textbf{z}_j\text{ ,}\quad  j=1,\dots,n
\end{equation}
where
$\textbf{x}_j\in \mathbb{C}^{K}$ is the vector of channel input symbols transmitted by the $K$ aircraft,
 $\textbf{z}_j\in \mathbb{C}^{M}$ is a complex Gaussian noise $\mathcal{CN}(0, N_0\textbf{I}_M)$, and $n$ is the codeword length. 
As explained before, the channel matrix $\textbf{H}$ is constant over the transmission of one codeword, i.e., during $n$ channel uses. In this work, we assume that $\textbf{H}$ can be perfectly estimated at the receiver.

\subsection{Capacity Region}
From an information-theoretic point of view, the channel with $K$ aircraft and one receiver, i.e., \ac{gs}, with $M$ receiving antennas and a fixed channel matrix \textbf{H} is a vector Gaussian \ac{mac}. An achievable rate $K$-tuple is a vector of rates $(R_1, R_2, \ldots, R_K)$ that can be simultaneously transmitted by the users with arbitrarily small probability of error \cite{Cover2006}. The capacity region $\mathcal{C}_{\text{MAC}}$ is the convex closure of the set of all achievable rates. From \cite[p. 425]{mumimoBook}, under the assumption that all $K$ aircraft transmit at a fixed power $P$, and that each receiver antenna has a noise power $N_0$ the capacity region is given by 
\begin{align}\label{eq:capacity}
\begin{split}
      \mathcal{C}_{\text{MAC}}& =\Bigg\{(R_1, R_2, \ldots, R_K):\sum_{k\in S} R_k \leq \text{\ldots   }\\   &\log_{2}\det\left(\textbf{I}_M+\frac{P}{N_0}\textbf{H}_{:S}\textbf{H}_{:S}^{\text{H}}\right), \forall S \subseteq \{1,2,\ldots,K\}\Bigg\} \text{ ,}
\end{split}
\end{align}
where $\textbf{H}_{:S}$ is the submatrix obtained from $\textbf{H}$ by taking the columns $\{\textbf{h}_{:k}: k \in S\}$ and the superscript H refers to the conjugate transpose. 
From the converse result on the capacity region and the arguments in \cite{ozarow}, for any coding scheme  transmitting at rates $\textbf{R}=(R_1, R_2, \ldots, R_K)$, if \textbf{H} is such that $\textbf{R}\notin \mathcal{C}_{\text{MAC}}$, then the probability that some aircraft messages are decoded in error is close to 1. 

For example, consider two arbitrary channel matrices with $K=2$, the corresponding rate regions for each channel matrix realization are shown in Fig.~\ref{fig:capacity}. In both rate regions, the rate $A$ is the achievable rate for aircraft $1$ in the absence of interference from aircraft $2$, and $B$ is the achievable rate for aircraft $2$ in the absence of interference from aircraft $1$. In contrast, $R_1^{2}$ and $R_2^{1}$ are the achievable rates for aircraft $1$ and $2$, respectively, when 
the decoder at the \ac{gs} is treating the interference as noise.

The pink, orange, and green areas are within the capacity region. 
In Fig.~\ref{fig:ex1}, the capacity region is highly symmetric, indicating that the achievable rates for both aircraft are quite similar. In contrast, Fig.~\ref{fig:ex2} shows a strongly asymmetric capacity region.  Such asymmetry typically results from a strong power imbalance in the signals received from different aircraft. In \ac{ag} communications, this situation is likely to occur due to large distance differences between the aircraft and the \ac{gs}. Therefore, cases with highly asymmetric capacity regions cannot be neglected.

Of particular interest are the so-called ``corner points" of the capacity region, corresponding to the points $(A,R_2^{1})$ and $(R_1^{2},B)$. These points are on the ``dominant face" of the capacity region and are therefore sum rate optimal. 
The following equality holds for these points
\begin{equation}\label{eq:K2}
    R_1^{2}+B=A+R_2^{1}=\log_{2}\det\left(\textbf{I}_M+\frac{P}{N_0}\textbf{H}\textbf{H}^{\text{H}}\right)\text{ .}
\end{equation}
It is well known that these corner rate points can be achieved by a particularly simple decoding scheme known as \ac{sic} decoding. For example, the point $(A,R_2^{1})$  can be achieved by first decoding user $2$ while treating the interference from user $1$ as noise, and then subtracting the codeword of user $2$ from the received signal. Since the probability of error of user $2$ at rate $R_2^{1}$ is vanishing, \ac{sic} can (almost) perfectly remove the interference, such that the decoder for user $1$ operates on a clean channel. Therefore, the rate $A$ can be achieved for user $1$, since at this point the decoder sees an interference-free channel. \rewtwo{On the other hand, to achieve any points on the line between the points $(A,R_2^{1})$ and $(R_1^{2},B)$, a joint group decoder must be used. Unlike \ac{sic}, which decodes users successively, joint group decoders process multiple users (aircraft) simultaneously. There are various approaches for this type of decoding. 
Representative examples include the widely cited works \cite{joint1,joint3,joint4,joint5}, among many others. }

\begin{figure}[!t]
\centering
\subfloat[]{\includegraphics[width=0.22\textwidth]{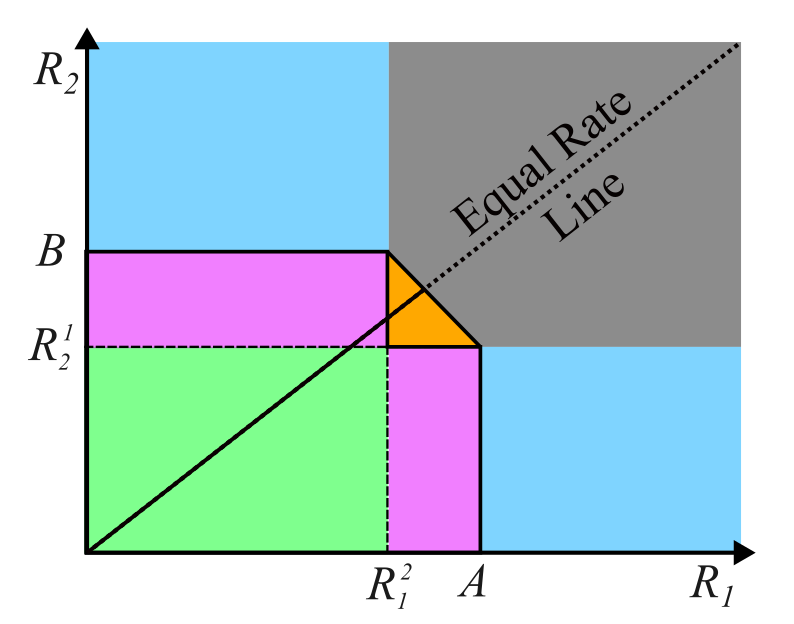}%
\label{fig:ex1}}
\hfil
\subfloat[]{\includegraphics[width=0.22\textwidth]{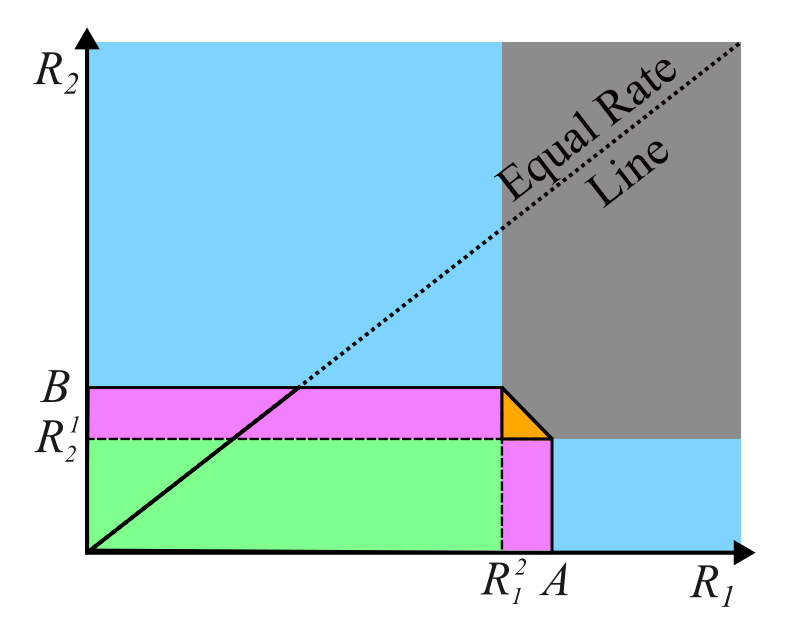}%
\label{fig:ex2}}
\caption{Example rate regions for $K=2$}
\label{fig:capacity}
\end{figure}

\subsection{Outage Event Identification}
Aircraft $k$ transmits at rate $r_k$ and the system's guaranteed rate is $r_G$. We consider two cases for $\forall k \in \{1,2,...,K\}$: the first is the \textit{equal-rate system}, where $r_k = r_G$, and the second is the \textit{variable-rate system}, where $r_k \geq r_G$. 
To determine which aircraft can reliably transmit, i.e., which aircraft are not in outage, we need to examine the capacity region. 
For example, when $K=2$, the rate point is denoted as $(r_1, r_2)$. In an equal-rate system, the rate point $(r_1, r_2)=(r_G, r_G)$ is on the equal-rate line  shown in Fig.~\ref{fig:capacity}.
On this line, the solid portion represents the regions within the capacity region, while the dashed portions indicate areas outside the capacity region. By examining the position of any given rate point on the rate region, we can determine which aircraft's transmitted signals can be successfully decoded. Additionally, the color of the region indicates the decoding technique required, such as \ac{isu} decoders, \ac{sic} decoders, or joint group decoders.
In particular:
\begin{itemize}
    \item In the green area, both signals can be decoded separately using \ac{isu} decoders by treating the interference {from the other user (aircraft)} as noise.
    \item Inside the pink areas, both signals can be decoded using a \ac{sic} scheme \cite{Cover2006}. 
    It is important to select the correct decoding order to ensure that both transmitted signals are successfully decoded. For example, if $R_2^{1}<r_2\leq B$ and $r_1\leq R_1^{2}$, then we must first decode aircraft $1$ and then aircraft $2$.
    \item Within the orange area, joint group decoding becomes necessary, where the transmitted signals are decoded together as a group. This process is more complex to implement in practice.
    \item In the blue areas, the rate vector is not in the capacity region. However, one of the transmitted signals can still be decoded by treating the other signal as noise. The specific blue area where the rate point is located determines which aircraft will be in outage. For example, if $B<r_2$ and $r_1\leq R_1^{2}$ then aircraft $2$ will be in outage but aircraft $1$ will not.
    \item Inside the gray area, both aircraft will be in outage.
\end{itemize}
\rewthree{As a result, if the rate point is outside the capacity region, this means that at least one aircraft is definitely in outage. However, this does not imply that all transmitting aircraft are in outage; some aircraft may still be successfully decoded.}

\subsection{Outage Probability}\label{subsec:outage_pro}
To calculate the probability of outage, 
we need to find the largest set of aircraft \mbox{$S^*\subseteq\{1,2,\ldots,K\}$} such that $\forall k\in S^*$ can be decoded at rate $r_k$.  Consequently, the complementary set, $\hat{S}\subseteq \{1,2,\ldots,K\}$, where $S^* \cap \hat{S}=\emptyset$, represents the set of aircraft in outage. Accordingly, the outage probability, $P_{\text{out}}$, is
\begin{equation}
    P_{\text{out}}=1-\frac{\mathbb{E}[|S^*|]}{K}=\frac{\mathbb{E}[|\hat{S}|]}{K}\text{ ,}
\end{equation}
where $\mathbb{E}[\cdot]$ denotes the expectation with respect to the realizations of the random channels \textbf{H}, and is evaluated by Monte Carlo simulation for the model defined in Section~\ref{sec:system_model}.

Our goal is to minimize the outage probability. This is accomplished by finding the set $S^*$ of maximum size for each realization of the channel matrix $\textbf{H}$. Hence, for a given rate point $(r_1,\ldots,r_K)$ and a given $\textbf{H}$, we wish to find
\begin{equation}\label{eq:probGSA}
    \max |S^*|\text{ .}
\end{equation}
In particular, it follows from \eqref{eq:capacity} that $S^*$ is the maximal set satisfying the condition: 
\begin{align}\label{eq:con_GSA}
\begin{split}
  \forall S \subseteq S^*,\quad\quad&\\\sum_{k\in S}r_k \leq &\log_{2}\det\Biggl(\textbf{I}_M+\ldots\\ &\quad\frac{P}{N_0}\textbf{H}_{:S}\textbf{H}_{:S}^{\text{H}} \left(\textbf{I}_M+\frac{P}{N_0}\textbf{H}_{:\hat{S}}\textbf{H}_{:\hat{S}}^{\text{H}}\right)^{-1} \Biggr)\text{ ,}
\end{split}
\end{align}
where $\textbf{H}_{:\hat{S}}$  is the matrix representing the aircraft that are in outage, with columns $\{\textbf{h}_{:k}:k\in \hat{S}\}$.  

Finding the maximal set  $S^*$  becomes more challenging as the number of aircraft in the system, $K$, increases.
A brute-force approach to derive $S^*$ requires considering every possible combination of $[K] = \{1, 2, \ldots, K\}$. 
Additionally, for each combination, we must evaluate all non-empty subsets, which amounts to $(2^v-1)$ subsets for a combination of $v$ aircraft. This results in  a total of
\begin{align}\label{eq:comp_out}
\begin{split}
    &\sum_{v=1}^K\sum_{S^* \in \binom{[K]}{v}}(2^v-1)\\
    &=\sum_{v=1}^K\frac{K!}{v!(K-v)!}(2^v-1)\text{ ,}
\end{split}
\end{align}
where we define  \mbox{$\binom{[K]}{v}$} as the set of all distinct subsets of cardinality $v$ of the set $[K]$. This total in \eqref{eq:comp_out} is the number of times the condition in \eqref{eq:con_GSA} must be evaluated in order to derive $S^*$.

\subsection{Outage Probability for SIC decoders}
The outage probability for \ac{sic} decoders means that the solution space is limited by the exclusion of joint group decoding. 
The largest set of non-outage aircraft under the \ac{sic} scheme is defined as $S^*_{\text{SIC}}$, where $S^*_{\text{SIC}}$ is a subset of $S^*$.
Accordingly, the outage probability for \ac{sic} decoders is
$P_{\text{out}}^{\text{SIC}}=1-\frac{\mathbb{E}[ |S^*_{\text{SIC}}|]}{K}$, where $P_{\text{out}}^{\text{SIC}}\geq P_{\text{out}}$.

Let  \mbox{$i_1\rightarrow i_2\rightarrow ...\rightarrow i_K$} be a permutation of the integers $1,2,...,K$ representing the decoding order. We define  $R_{i_u}^{T_u}$ as
\begin{align}\label{eq:sic}
\begin{split}
     R_{i_u}^{T_u}=\log_{2}&\det\Biggl(\textbf{I}_M+\ldots\\
     &\frac{P}{N_0}\textbf{h}_{:i_u}\textbf{h}_{:i_u}^{\text{H}} \left(\textbf{I}_M+\frac{P}{N_0}\textbf{H}_{:T_u}\textbf{H}_{:T_u}^{\text{H}}\right)^{-1} \Biggr)\text{ ,}
\end{split}
\end{align}
where $\textbf{H}_{:T_u}$ denotes the matrix formed by the columns \mbox{$\{\textbf{h}_{:i_k}:k>u\}$} and $u\in \{1,2,\ldots,K\}$. 
In this context, we define the set $S^*_{\text{SIC}}$ as  
\begin{align}\label{eq:condition}
    \begin{split}
         S^*_{\text{SIC}} = \Big\{ i_1, &\dots,i_{u^* - 1} :\dots\\ &u^* = \min_{u \in \{1,\ldots,K\}} u \text{ subject to } r_{i_u} > R_{i_u}^{T_u} \Big\}\text{ ,}
    \end{split}
\end{align}
where $u^*$ is the smallest index for which the outage condition, \mbox{$r_{i_{u^*}} > R_{i_{u^*}}^{T_{u^*}}$}, holds. If none of the aircraft are in outage, i.e., \mbox{$r_{i_u}\leq R_{i_u}^{T_u}, \forall u \in \{1, \dots,K\}$}, then we assume that \mbox{$\min_{u \in \{1,\ldots,K\}} u =K+1$} in \eqref{eq:condition}, such that the cardinality of the set $S^*_{\text{SIC}}$ is $K$.
Thus, the size of the $S^*_{\text{SIC}}$ depends on the decoding order. In order to minimize the outage probability, the optimization problem is formulated as 
\begin{equation}\label{eq:probSSA}
    \max_{\Pi} |S^*_{\text{SIC}}|\text{ ,}
\end{equation}
where $\Pi$ represents the decoding order. 
A brute-force approach to find an optimal decoding order that maximizes the size of $S^*_{\text{SIC}}$ requires evaluating all $K!$ permutations, which becomes computationally very expensive as $K$ increases.

In the rest of the paper, for arbitrary values of $x$ and $y$, the notation $R_x^{y}$ means the achievable rate of $x$ under the interference of $y$, and it is calculated by the \ac{rhs} of \eqref{eq:sic} where $i_u$ is defined by $x$ and $T_u$ is defined by~$y$. 

\begin{lemma}\label{lemma:ssa}
Let $T_o$ and $T_p$ be two distinct sets, where $o\not\in T_o$ and $p\not\in T_p$. If $o=p$ and $T_o\subset T_p$, then it follows that $R_{o}^{T_o}\geq R_{p}^{T_p}$ \cite{vblast}\text{ .}
\end{lemma}
\begin{proof}
The sets $T_o$ and $T_p$ represent the interfering aircraft with $o$ and $p$, respectively. According to Shannon's capacity equation, an increase in interference leads to lower achievable data rate. 
As a result, as more aircraft interfere with the transmitted signal, the achievable data rate will either remain the same or decrease, depending on the linear dependence of the channel vectors associated with the aircraft.  
\end{proof}

\subsection{Outage Probability for ISU decoders}
\ac{isu} decoders decode the signals from each aircraft separately. For example, well-known \ac{zf} and \ac{mmse} decoders are examples of \ac{isu} decoders \cite{mumimoBook}. When decoding aircraft $k$, an \ac{isu} decoder treats the interference from the other aircraft as noise. 
Accordingly, the non-outage condition for aircraft $k$ is
\begin{align}\label{eq:isu}
\begin{split}
     r_k\leq \log_{2}&\det\Biggl(\textbf{I}_M+\frac{P}{N_0}\textbf{h}_{:k}\textbf{h}_{:k}^{\text{H}} \left(\textbf{I}_M+\frac{P}{N_0}(\textbf{H})_k(\textbf{H})_k^{\text{H}}\right)^{-1} \Biggr)\text{ ,}
\end{split}
\end{align}
where the notation $(\textbf{H})_k$ represents the removal of the $k$-th column from $\textbf{H}$. 
The largest set of non-outage users using \ac{isu} decoders is defined as $S^*_{\text{ISU}}$. Accordingly, the upper-bound on the minimal outage probability for \ac{isu} decoders is \mbox{$P_{\text{out}}^{\text{ISU}}=1-\frac{\mathbb{E}[|S^*_{\text{ISU}}|]}{K}$}, where \mbox{$P_{\text{out}}^{\text{ISU}}\geq P_{\text{out}}^{\text{SIC}}\geq P_{\text{out}}$}.

%% file: Sections/05.tex
%

\section{Proposed Algorithms for Calculating Outage Probability}\label{sec:algs}

We have shown in the previous section that the optimization problems given in \eqref{eq:probGSA} and \eqref{eq:probSSA} are difficult to solve with a brute force approach due to the associated high computational complexity. In this context, we propose three algorithms designed to efficiently identify  solutions while significantly reducing the computational complexity.
These algorithms not only compute {an upper bound on the minimal outage probability}, but also provide practical achievability strategies. As such, they provide explicit decoding orders for \ac{sic} and \ac{sic} with joint group detection in subgroups that can be implemented in a practical receiver.  
{The first algorithm is \ac{ssa}, this aims to solve the problem given in \eqref{eq:probSSA}. \ac{ssa} identifies the largest subset of decodable aircraft by excluding the joint group decoding options. 
The second algorithm is \ac{gsa}, which extends \ac{ssa} by using its outputs as inputs. \ac{gsa} solves the problem in \eqref{eq:probGSA}, where joint group decoding options are included.} The third algorithm is \ac{lgsa}, which is a modified version of \ac{gsa}. With \ac{lgsa} we find a suboptimal solution to the problem in \eqref{eq:probGSA} by limiting the number of aircraft that can be decoded jointly.

The complexity of the problems given in \eqref{eq:probGSA} and \eqref{eq:probSSA} is reduced by starting the evaluation from the smallest subsets and using the insights gained from these evaluations. This allows us to eliminate certain subset options or decoding orders from further consideration, thereby narrowing the search space and improving computational efficiency. 
There are three sets in the algorithms: $S^*$, $\hat{S}$, and $L$. These three sets form a partition of \mbox{$\{1,2,\ldots,K\}$}.
The set $S^*$ contains aircraft that can be successfully decoded at their transmission rate, i.e., \mbox{$\forall k \in S^*$} can be decoded at the rate $r_k$. Set $\hat{S}$ contains the aircraft that are certainly in outage and $L$ contains the aircraft whose status has not yet been determined.   As the algorithms progress, the aircraft that can be decoded successfully or are known to be in outage are removed from $L$ and added to either $S^*$ or $\hat{S}$. Thus, the objective of the algorithms is to minimize the set $L$ by removing as many elements as possible. The aircraft that cannot be removed from $L$ at the end of the algorithms can be considered as being in outage, since no solution could be found to successfully decode them.


\subsection{Single Successive Algorithm (SSA)}
{\ac{ssa} is presented in Algorithm~\ref{al:ssa} and solves the optimization problem given in \eqref{eq:probSSA}.} Please note that $S^*_{\text{SIC}}$ in \eqref{eq:probSSA} is denoted by $S^*$ in this algorithm.
We start by setting \mbox{$L=\{1,2,...,K\}$}, \mbox{$S^*=\emptyset$}, and \mbox{$\hat{S}=\emptyset$}, and then move the aircraft from $L$ to either $S^*$ or $\hat{S}$ during two main phases.

The first phase is referred to as  \textit{PruneAircraft}. It identifies aircraft that are definitely in outage.
Consider $l\in L$,  in line~\ref{ssa:line:outage}, $R_{l}^{\hat{S}}$ is defined as the achievable rate of aircraft $l$ under the interference of the outage set, $\hat{S}$.
If \mbox{$R_{l}^{\hat{S}}< r_l$}, then aircraft $l$ is in outage.
This is because the interference caused by the aircraft in $\hat{S}$ cannot be removed from the received signal by any \ac{sic} decoding strategy.
Therefore, the achievable rate for aircraft $l$ will always be less than or equal to $R_{l}^{\hat{S}}$. Lemma~\ref{lemma:ssa} supports this argument, since it indicates that in order to achieve higher rates, the interference on the target signal should be reduced. 
However, in this case, it is not possible to remove the interference caused by the aircraft already assigned to the outage set $\hat{S}$.
As a result, aircraft $l$ is removed from $L$ and added to $\hat{S}$.
The \textit{PruneAircraft} phase ends when the algorithm can no longer find a new aircraft to add to  $\hat{S}$. The search space for the next step becomes smaller after this phase, since the number of the iterations in the next step depends on the size of $L$.



\begin{algorithm}[t!]
\centering
\caption{Single Successive Algorithm (SSA)}\label{al:ssa}
\small
\begin{algorithmic}[1]

\Require{$\textbf{H}$, $\textbf{r}=(r_1,\ldots,r_K)$}
\Ensure $S^*$, $\hat{S}$, $L$

\Algphase{\textbf{Initialization}}
\State $L=\{1,2,...,K\}$, $S^*=\emptyset$, $\hat{S}=\emptyset$
\Algphase{\textbf{Phase 1 - Prune Aircraft}}
\Function{PruneAircraft}{$\textbf{H}$, $\textbf{r}$, $L$, $\hat{S}$}
\While{$|L|>0$}
    \State tmp=$\hat{S}$
        \ForEach{$l \in L$}
             \State $ \begin{aligned}[t] R_l^{\hat{S}}=\log_{2}\det\Biggl(&\textbf{I}_M+\ldots\\ &\frac{P}{N_0}\textbf{h}_{:l}\textbf{h}_{:l}^{\text{H}} \left(\textbf{I}_M+\frac{P}{N_0}\textbf{H}_{:\hat{S}}\textbf{H}_{:\hat{S}}^{\text{H}}\right)^{-1} \Biggr)\end{aligned}$ \label{ssa:line:outage}
            \If{$r_{l} >R_l^{\hat{S}}$}
                \State Remove $l$ from $L$
                \State $\hat{S} \gets \text{add } l$
            \EndIf
        \EndForEach
    \State \textbf{if} isequal(tmp,$\hat{S}$) \textbf{then break} 
\EndWhile
\State \Return  $L$, $\hat{S}$
\EndFunction
\Algphase{\textbf{Phase 2 - Greedy SIC}}
\Function{GreedySIC}{$\textbf{H}$, $\textbf{r}$, $L$,  $S^*$, $\hat{S}$}
\While{$|L|>0$}
    \State tmp=$S^*$
        \ForEach{$l \in L$}
            \State $T_l=(L\cup\hat{S})\symbol{92}\{l\}$ \label{ssa:line:Tl}
            \State $ \begin{aligned}[t] R_l^{T_l}=\log_{2}\det&\Biggl(\textbf{I}_M+\ldots\\ &\frac{P}{N_0}\textbf{h}_{:l}\textbf{h}_{:l}^{\text{H}} \left(\textbf{I}_M+\frac{P}{N_0}\textbf{H}_{:T_l}\textbf{H}_{:T_l}^{\text{H}}\right)^{-1} \Biggr)\end{aligned}$ \label{ssa:line:constraint}
            \If{$r_{l} \leq R_l^{T_l}$}
                \State Remove $l$ from $L$
                \State $S^*\gets \text{add } l$
            \EndIf
        \EndForEach
    \State \textbf{if} isequal(tmp,$S^*$) \textbf{then break} 
\EndWhile
\State \Return  $L$, $S^*$
\EndFunction
\end{algorithmic}
\end{algorithm}

The second phase is referred to as \textit{GreedySIC}. This step finds the aircraft that can be successfully decoded using \ac{sic} decoders.
In line~\ref{ssa:line:Tl} of Algorithm~\ref{al:ssa}, \mbox{$T_l=L\cup\hat{S} \symbol{92}\{l\}$}  is defined as the constraint set of aircraft $l$ and it contains all aircraft in $L$ and $\hat{S}$ except $l$. Then, in line~\ref{ssa:line:constraint} of Algorithm~\ref{al:ssa}, $R_{l}^{T_l}$ is defined as the achievable rate of aircraft $l$ under the interference of $T_l$.
Suppose that the \ac{sic} decoder decodes aircraft $l$ treating the aircraft not yet decoded, i.e., \mbox{$T_l$}, as interference.  If \mbox{$r_l\leq R_{l}^{T_l}$}, then the decoding of aircraft $l$ is successful and the interference caused by aircraft $l$ can be removed from the received signal. Hence, in the subsequent steps, 
we remove $l$  from $L$ and add it to $S^*$.
This shrinks the constraint set of the remaining aircraft in $L$. As a result, by Lemma~\ref{lemma:ssa}, the probability of a successful decoding of the remaining aircraft in $L$ increases. Therefore, each time a new aircraft is added to $S^*$, we retry to decode the remaining aircraft in $L$ under the interference of the shrunk constraint set.
This phase ends when all aircraft in the set $L$ hold the condition \mbox{$R_{l}^{T_l}<r_l$}. This means that there is no \ac{sic} decoding order that can successfully decode any of the aircraft in $L$.
 
At the end of Algorithm~\ref{al:ssa}, we know that the aircraft in $\hat{S}$ are definitely in outage.
The aircraft in $L$  cannot be decoded using only \ac{sic} decoders, but they may be decoded when joint group decoders are applied. As a result, the aircraft in $S^*$ form the largest possible $S^*_{\text{SIC}}$. {In addition to identifying the elements of optimized $S^*_{\text{SIC}}$, we also find an optimal decoding order for the problem in \eqref{eq:probSSA}.} 



\subsection{Group Successive Algorithm (GSA)}
\ac{gsa} is presented in Algorithm~\ref{al:step2} and solves the optimization problem given in \eqref{eq:probGSA} by considering the joint group decoding options in addition to the \ac{sic} scheme.  Algorithm~\ref{al:step2} (\ac{gsa}) is an enhancement of Algorithm~\ref{al:ssa} (\ac{ssa}). In particular, Algorithm~\ref{al:step2} uses the outputs of Algorithm~\ref{al:ssa}  which are: $L$, $S^*$, and $\hat{S}$. 
The set $S^*$ contains the aircraft that were successfully decoded in the \ac{ssa}; we assume that their influence on the received signal has been removed and therefore we do not consider these aircraft any further. $\hat{S}$  is the set of aircraft that are certainly in outage. Consequently, we always have to consider the interference caused by these aircraft, but we do not try to decode them. This further narrows the search space. 
The third set, $L$, contains the aircraft for which the following condition holds: 
\mbox{$\forall l\in L$, $R_l^{T_l}<r_l\leq R_l^{\hat{S}}$}. 
This means they can achieve the rate $r_l$ under the interference of outage set, $\hat{S}$, but cannot achieve this rate under the interference of their individual constraint sets, $T_l$. 
We try to determine if these aircraft can be decoded jointly as a group, or if they are definitely in outage. 
In the rest of Algorithm~\ref{al:step2}, we define \mbox{$\binom{L}{v}$} as the set of all distinct subsets of cardinality $v$ drawn from $L$. We also introduce a new set  $C$ to represent any combination of $v$ aircraft from $L$, such that \mbox{$C \in \binom{L}{v}$} and hence  \mbox{$|C|=v$}.


\begin{algorithm}[t!]
\small
\centering
\caption{Group Successive Algorithm (GSA)}\label{al:step2}
\begin{algorithmic}[1]
\Require{$\textbf{H}$, $\textbf{r}=(r_1,\ldots,r_K)$}
\Ensure $S^*$, $\hat{S}$,$L$
\Algphase{\textbf{Phase 1 - Single Successive Algorithm}}
\State $L=\{1,2,...,K\}$, $S^*=\emptyset$, $\hat{S}=\emptyset$
\State $L$, $\hat{S}$ $\gets$ \Call{PruneAircraft}{$\textbf{H}$, $\textbf{r}$, $L$, $\hat{S}$}
\State $L$, $S^*$ $\gets$ \Call{GreedySIC}{$\textbf{H}$, $\textbf{r}$, $L$,  $S^*$, $\hat{S}$}
\Algphase{\textbf{Phase 2 - Prune Subsets}}
\State \rewtwo{$q=2$}\label{gsa:line:q2}
\While{$|L|\geq q$ \rewtwo{\textbf{and} $q_{\text{max}}\geq q$} }
    \State tmp=$\hat{S}$
        \ForEach{$C \in\binom{L}{\rewtwo{q}}$}
            \State $ \begin{aligned}[t] R_C^{\hat{S}}=\log_{2}\det\Biggl(&\textbf{I}_M+\ldots\\ &\frac{P}{N_0}\textbf{h}_{:C}\textbf{h}_{:C}^{\text{H}} \left(\textbf{I}_M+\frac{P}{N_0}\textbf{H}_{:\hat{S}}\textbf{H}_{:\hat{S}}^{\text{H}}\right)^{-1} \Biggr)\end{aligned}$\label{gsa:line:outage}
            \If{$\sum_{k\in C} r_{k} > R_C^{\hat{S}}$}
                \State Remove $C$ from $L$ \label{gsa:line:ph1out}
                \State $\hat{S} \gets \text{add } C$
                \State \textbf{break} \label{gsa:line:pBreak}
            \EndIf
        \EndForEach
    \rewtwo{\If{\textbf{not} isequal(tmp,$\hat{S}$)}
    \State $q=1$\label{gsa:line:pRepeat}
    \Else
    \State $q=q+1$ \label{gsa:line:qincrement}
    \EndIf}
\EndWhile
\Algphase{\textbf{Phase 3 - Greedy Group}}
\State $v=2$ \label{gsa:line:v2}
\While{$|L|\geq v$}\label{gsa:line:while}
    \State tmp=$L$
   
        \ForEach{$C \in \binom{L}{v}$}
            \State $T_C=(L\cup\hat{S})\symbol{92}\{C\}$\label{gsa:line:TC}
            
            \IfT $\begin{aligned}[t] \forall  S\subseteq C, \sum_{k\in S} &r_{k} \leq  R_{C,S}^{T_C}=\log_{2}\det\Biggl(\textbf{I}_M+...\\
               &\frac{P}{N_0}\textbf{H}_{:S}\textbf{H}_{:S}^{\text{H}} \left(\textbf{I}_M+\frac{P}{N_0}\textbf{H}_{:T_C}\textbf{H}_{:T_C}^{\text{H}}\right)^{-1} \Biggr)
           \end{aligned}$\label{gsa:line:constraint}
                    \State Remove $C$ from $L$
                    \State $S^* \gets \text{add } C$
                    \State \textbf{break} \label{gsa:line:break}
            \EndIfT
        \EndForEach
    \If{\textbf{not} isequal(tmp,$L$)}
    \State $v=1$ \label{gsa:line:v1}
    \Else
    \State $v=v+1$ \label{gsa:line:vPlus}
    \EndIf
\EndWhile
\end{algorithmic}
\end{algorithm}

We refer to the next phase of Algorithm~\ref{al:step2} as \textit{PruneSubsets}. It  identifies the aircraft that are definitely in outage, thereby minimizing the size of $L$.  
\rewtwo{In this phase, the algorithm examines all subsets of $L$ whose cardinalities range from $1$ up to $q_{\text{max}}$.
For given \mbox{$C\subseteq L$}, where $|C|=q$, $R_C^{\hat{S}}$ is defined in line~\ref{gsa:line:outage} of Algorithm~\ref{al:step2}.} Hence, $R_C^{\hat{S}}$ is the achievable sum rate of aircraft in $C$ under the interference of the outage set, $\hat{S}$.
\begin{lemma}\label{lemma:gsa}
\rewtwo{Let $C=\{o,p, \dots, w\}$ be a set of aircraft, where $2\leq |C|=q$, each transmitting at rates $\{r_o, r_p, \dots, r_{w}\}$, respectively. $\hat{S}$ is the outage set, i.e., users in $\hat{S}$ are certainly in outage. Suppose  \mbox{$\forall S \subset C$} with \mbox{$1\leq|S|< |C|$}, we have \mbox{$\sum_{k\in S}r_k \leq R_S^{\hat{S}}$}  and, at the same time,  \mbox{$\sum_{k\in C} r_k >R_C^{\hat{S}}$}. Then all aircraft in $C$ are in outage.}
\end{lemma}
\begin{proof}
\rewtwo{We define $S^{c}$ as the complementary set of $S$ such that $S$ and $S^{c}$ form a partition of $C$, i.e., \mbox{$S\cup S^c=C$}  and \mbox{$S\cap S^c=\emptyset$}. We also define $T_{S^{c}}=S\cup\hat{S}$.
Accordingly, from \eqref{eq:capacity}, we have 
\begin{equation}\label{eq:221}
    R_C^{\hat{S}}=R_{S^{c}}^{T_{S^{c}}}+R_S^{\hat{S}}\text{ ,}
\end{equation}
where $R_{S^{c}}^{T_{S^{c}}}$ denotes the achievable rate of the users in $S^{c}$ under the interference of $T_{S^{c}}$, and $R_S^{\hat{S}}$ denotes the achievable rate of the users in $S$ under the interference of $\hat{S}$ only.
It is given that $\sum_{k\in C} r_k > R_C^{\hat{S}}$.
Substituting into \eqref{eq:221}, we obtain
\begin{equation}\label{eq:222}
    \sum_{k\in C} r_k>R_{S^{c}}^{T_{S^{c}}}+R_S^{\hat{S}}\text{ .}
\end{equation}
Expressing \eqref{eq:222} as sums over the partitions, we have
\begin{equation}\label{eq:223}
    \sum_{k\in S^c} r_k + \sum_{k\in S} r_k  >R_{S^{c}}^{T_{S^{c}}}+R_S^{\hat{S}}\text{ .}
\end{equation}
Since, it is also given that for any subset \mbox{$S \subset C$}, where \mbox{$|S|< |C|$}, we have \mbox{$\sum_{k\in S} r_k \leq R_S^{\hat{S}}$}, we can replace \mbox{$\sum_{k\in S} r_k$} with $R_S^{\hat{S}}$ in the inequality to derive
\begin{equation}\label{eq:224}
    \sum_{k\in S^c} r_k + R_S^{\hat{S}} >R_{S^{c}}^{T_{S^{c}}}+R_S^{\hat{S}}\text{ ,}
\end{equation}
which simplifies to
\begin{equation}\label{eq:225}
    \sum_{k\in S^c} r_k >R_{S^{c}}^{T_{S^{c}}}\text{ .}
\end{equation}
This means $\forall S^c \subset C$, the aircraft in $S^c$ cannot achieve their transmission rates under the interference of aircraft in $S$.
Since $\sum_{k\in C} r_k > R_C^{\hat{S}}$, joint group decoding of the users in $C$ is not possible.
As a result, there is no subset of $C$ that can be successfully decoded. Therefore, all aircraft in $C$ are definitely in outage.}
\end{proof}
\rewtwo{Using Lemma~\ref{lemma:gsa}, we identify more aircraft that are certainly in outage. We know that all aircraft in $L$ satisfy \mbox{$\forall l \in L, r_l\leq R_l^{\hat{S}}$}, since they were  not moved to $\hat{S}$ during the \textit{PruneAircraft} phase and the set $\hat{S}$ has not change since then. This corresponds to all size 1 combinations of $L$, i.e., $\binom{L}{q}$ with $q=1$. Therefore, we next evaluate the subsets of size 2, setting $q=2$ in line~\ref{gsa:line:q2} of Algorithm~\ref{al:step2}. During the \textit{PruneSubsets} phase, the value of $q$ is updated based on the following events:}
\begin{enumerate}
    \item \rewtwo{If $\forall C \in \binom{L}{q}, \sum_{k\in C} r_k\leq R_C^{\hat{S}}$, then we increment the  value of $q$ by one (see line~\ref{gsa:line:qincrement}).}
    \item \rewtwo{If any $C\in \binom{L}{q}$ fails to satisfy the condition \mbox{$\sum_{k\in C} r_k \leq R_C^{\hat{S}}$}, then  the  aircraft in $C$ are in outage by Lemma~\ref{lemma:gsa}. 
    This is because we know that aircraft in $L$ with cardinalities from 1 up to $q-1$ satisfy their respective conditions, i.e.,  \mbox{$\forall S\subset L, 1\leq |S| \leq q-1, \sum_{k\in S} r_k \leq R_S^{\hat{S}}$}. 
    Accordingly, the aircraft in $C$ are removed from $L$ and added to $\hat{S}$ (see line~\ref{gsa:line:ph1out}). As the number of aircraft in $\hat{S}$ increases, some aircraft in 
    $L$ may no longer satisfy their respective conditions. For example, for an aircraft $l\in L$, the condition  $r_l\leq R_l^{\hat{S}}$ might no longer hold. Therefore, subsets of $L$ starting from $q=1$ must be re-evaluated. Each time new aircraft are added to $\hat{S}$, the ongoing \textit{for-loop} is interrupted (see line~\ref{gsa:line:pBreak}) and $q$ is reset to 1 (see line~\ref{gsa:line:pRepeat}).}  
\end{enumerate}
\rewtwo{The \textit{PruneSubsets} phase ends when either \mbox{$q > q_{\text{max}}$} or no aircraft remain in $L$.} 
\rewtwo{By reducing the number of aircraft in $L$, this  phase significantly lowers the computational complexity of the subsequent phase.
The overall computational complexity of Algorithm~\ref{al:step2} depends on the value of $q_{\text{max}}$. The choice of $q_{\text{max}}$ that minimizes computational complexity is contingent upon $K$, $M$, and the transmission rates of the aircraft.
In this paper, we set $q_{\text{max}}=2$, with the rationale provided in Appendix~\ref{ap:q_max}.
}

The final phase of Algorithm~\ref{al:step2} is called \textit{GreedyGroup}. In this step, the remaining aircraft in $L$ are evaluated. 
The set $C$ represents one combination of $v$ aircraft in $L$, i.e., \mbox{$C\in \binom{L}{v}$}, and hence \mbox{$|C|=v$}.  Initially, $v$ is set to $2$ in line~\ref{gsa:line:v2} of Algorithm~\ref{al:step2}, the value of $v$ changes during the course of the function.
In line~\ref{gsa:line:TC} of Algorithm~\ref{al:step2}, \mbox{$T_C=(L\cup\hat{S})\symbol{92}\{C\}$} is defined as the constraint set of $C$ and contains all aircraft in $L$ and $\hat{S}$ except $C$. $R_{C,S}^{T_C}$ is defined in line~\ref{gsa:line:constraint} of Algorithm~\ref{al:step2} as the achievable sum rate of aircraft in $S$, which is a subset of $C$, under the interference of $T_C$. 
For a given value of $v$ and $L$, there are two possibilities:
\begin{enumerate}
    \item
    If \mbox{$\forall S\subseteq C, \sum_{l\in S}r_l\leq R_{C,S}^{T_C} $}, then the aircraft in $C$ can be decoded successfully, and their interference can be cancelled. 
    To this end, we remove the aircraft in $C$ from $L$ and add them to $S^*$.
    It is important to note that this condition is equivalent to the one given in \eqref{eq:con_GSA}, where $C$ is defined by $S^*$ and $T_C$ is defined by $\hat{S}$.
    Moreover, by Lemma~\ref{lemma:ssa}, removing aircraft in  $C$ from $L$ increases the probability of successful decoding for the remaining aircraft in $L$, as it reduces the number of aircraft in their individual constraint set.
    Therefore, each time some aircraft are removed from $L$, we break the \textit{for-loop} on line~\ref{gsa:line:break} of Algorithm~\ref{al:step2}  and change the value of $v$ to 1, as shown on line~\ref{gsa:line:v1}. Setting $v$ to 1 corresponds to considering a \ac{sic} decoder. This change in the value of $v$ is made because, with the shrunken constraint set, some aircraft in  $L$ may now be decoded using the \ac{sic} decoder. 
\item
    If none of the combinations of $v$ aircraft in $L$ satisfy the condition \mbox{$\forall S\subseteq C, \sum_{l\in S}r_l\leq R_{C,S}^{T_C} $}, i.e., there is no group $C$ with size $v$ in $L$ that can be successfully decoded, then we increment the value of $v$ by one, as shown on line~\ref{gsa:line:vPlus} of Algorithm~\ref{al:step2}. Increasing the value of $v$ means that we now consider decoding the aircraft in $L$ in a larger group.
\end{enumerate}
In summary, the value of $v$ changes in two scenarios during the \textit{GreedyGroup} phase: First, whenever some aircraft are removed from $L$ and added to $S^*$, then $v$ is set to 1. Second, if we cannot find any group of aircraft that can be successfully decoded at the current value of  $v$, then $v$ is incremented by 1, i.e., \mbox{$v=v+1$}.
The \textit{GreedyGroup} phase ends either when $|L|=0$, meaning all $K$ aircraft have been already assigned to either $S^*$ or $\hat{S}$, or when $|L|<v$, indicating that the aircraft in $L$ are certainly in outage, along with the aircraft in $\hat{S}$. With Algorithm~\ref{al:step2}, we find  the optimal set $S^*$ for the problem given in \eqref{eq:probGSA}. The proof is given in Appendix \ref{ap:gsa}.

\subsection{Limited Group Successive Algorithm (LGSA)}
The implementation of joint group decoders comes at the cost of high complexity, and as a result, joint group decoding of large groups of aircraft may not be feasible in real-world applications for the time being. \ac{lgsa} addresses this by limiting the number of aircraft that can be decoded jointly, which is achieved by a  modification of Algorithm~\ref{al:step2} (\ac{gsa}). In \ac{lgsa}, the  \textit{while-loop} condition in line~\ref{gsa:line:while} of Algorithm~\ref{al:step2} is modified from $|L|\geq v$ to $v_{\text{max}}\geq v$. This change forces the algorithm to ensure that group sizes remain less than or equal to $v_{\text{max}}$.
While this reduces the complexity, it comes at the expense of a higher outage probability.
\rewfour{ Determining a feasible value for   $v_{\text{max}}$ requires consideration of the available hardware capabilities and the specific implementation of joint group decoding.  As this investigation is beyond the scope of this paper, we instead present representative examples for $v_{\text{max}}\in \{2,4\}$  in Section~\ref{sec:results}.}

%% file: Sections/06.tex
\section{Literature Review of Existing Decoding Ordering Strategies}\label{sec:receiver}
The outage probability of \ac{sic} decoders depends on the decoding order, as explained in Section~\ref{sec:capacity}. 
To this end, in this section we present two decoding order strategies from the literature and compare their performance with \ac{ssa} in Section~\ref{sec:results}. The first is \mbox{V-BLAST}~\cite{vblast}, which is a prominent decoding order in the literature.  The second is the decoding order proposed in~\cite{order2}, which is based on the \ac{cgtr} of the aircraft. 

The set of non-outage aircraft for a given decoding order,  \mbox{$i_1\rightarrow i_2\rightarrow \dots \rightarrow i_K$}, is defined as $S^*_{\text{DO}}$.
Algorithm \ref{alg:outage_condition} shows how to determine $S^*_{\text{DO}}$.
\begin{algorithm}
\small
\caption{Non-Outage Aircraft for a Given Decoding Order}\label{alg:outage_condition}
\begin{algorithmic}[1]
\State $\hat{S}=\emptyset$
\For{$u=1,2,...,K$}
\State $F_u=\hat{S}\cup \{i_{u+1},\dots,i_K\}$\label{sic:line:T}
\If{$r_{i_u} \leq R_{i_u}^{F_u} $}
\State $S^*_{\text{DO}}\gets \text{add } i_u$
\Else
\State $\hat{S}\gets \text{add } i_u$
\EndIf
\EndFor
\end{algorithmic}
\end{algorithm}
In line \ref{sic:line:T} of Algorithm \ref{alg:outage_condition}, the aircraft that interferes with the $i_u$-th aircraft is defined as $F_u$.
The achievable rate for the aircraft $i_u$, $R_{i_u}^{F_u}$, is obtained using \eqref{eq:sic}, where $F_u$ corresponds to $T_u$. By simulating the aircraft positions multiple times, we compute the outage probability, $P_{\text{out}}^{\text{DO}}=1-\frac{\mathbb{E}[|S^*_{\text{DO}}|]}{K}$.

\subsection{V-BLAST}
\mbox{V-BLAST} maximizes the minimum \ac{sinr} at the detector output, by decoding the signal with the highest \ac{sinr} at each iteration \cite{vblast}. 
\begin{lemma}\label{lemma:proof}
If $r_k=r_G, \forall k \in \{1,2,\dots,K\}$, then an optimal decoding order of the optimization problem given in \eqref{eq:probSSA} is decoding the aircraft with the highest \ac{sinr} at each step of the \ac{sic} process.
\end{lemma}
\begin{proof}
Let  $o$ and $p$ be any two aircraft and $T_{o,p}$ be the set of all the other aircraft that have not yet been successfully decoded. Accordingly, we define $T$, $T_{o}$, $T_{p}$ as follows
\begin{align}
    \begin{split}
        &T=\{o,p\}\cup T_{o,p} \text{,}\\
        &T_o=\{p\}\cup T_{o,p}\text{,}\\
        &T_p=\{o\}\cup T_{o,p}\text{.}
    \end{split}
\end{align}

By Lemma~\ref{lemma:ssa}, $R_p^{T_{o,p}}\geq R_p^{T_p}$ and $R_o^{T_{o,p}}\geq R_o^{T_o}$, since $T_{o,p}$ is a subset of $T_p$ and $T_o$.
We consider two decoding orders:
\begin{itemize}
     \item \textit{Decoding order 1:} $p\rightarrow o\rightarrow T_{o,p}$ \\
    The achievable rate for aircraft $p$, $R_p^{T_p}$ , and in the case of successful decoding of aircraft $p$, the achievable rate for aircraft $o$ is  $R_o^{T_{o,p}}$.
    \item \textit{Decoding order 2:} $o\rightarrow p\rightarrow T_{o,p}$ \\
    The achievable rate for aircraft $o$, $R_o^{T_o}$, and in the case of successful decoding of aircraft $o$, the achievable rate for aircraft $p$ is $R_p^{T_{o,p}}$.
\end{itemize}
Let us assume that $R_o^{T_o}\geq R_p^{T_p}\geq R_k^{T_k},\forall k\in T_{o,p}$ where \mbox{$T_k=T\symbol{92}\{k\}$}, which contains all elements of $T$ except $k$.
This means that aircraft $o$ has the highest \ac{sinr} followed by $p$.
There are three possibilities:
\begin{itemize}
    \item \textit{Possibility 1: }$R_o^{T_o}\geq R_p^{T_p}\geq r_G$\\
    Both decoding orders can successfully decode aircraft $o$ and $p$, since $R_p^{T_{o,p}}\geq R_p^{T_p}$ and $R_o^{T_{o,p}}\geq R_o^{T_o}$. As a result, $\{o,p\}\in S^*_{\text{MMSE-SIC}}$.
    \item \textit{Possibility 2: }$r_G \geq R_o^{T_o}\geq R_p^{T_p}$\\
    Neither of the decoding orders can decode the first aircraft that needs to be decoded. Moreover, given that $R_o^{T_o}\geq R_p^{T_p}\geq R_k^{T_k},\forall k\in T_{o,p}$  no other aircraft in $T$ can be successfully decoded with any other decoding order.
    \item \textit{Possibility 3: }$R_o^{T_o}\geq r_G\geq  R_p^{T_p}$\\
    \begin{itemize}
        \item \textit{Decoding order 1} cannot decode aircraft $p$, so the interference caused by $p$ cannot be removed. Aircraft $o$ can be decoded under the interference of $p$. All the aircraft in $T_{o,p}$ must be decoded under the interference of $p$, which reduces their chances of successful decoding.
        \item \textit{Decoding order 2} can decode aircraft $o$. Since $R_p^{T_{o,p}}\geq R_p^{T_p}$, there is a chance that aircraft $p$ might also be successfully decoded. If $r_G\geq R_p^{T_{o,p}}\geq R_p^{T_p}$, then the  results will be the same as \mbox{\textit{decoding order 1}}. However, if $R_p^{T_{o,p}}\geq r_G\geq  R_p^{T_p}$, then aircraft $p$ can be successfully decoded, and the aircraft in $T_{o,p}$ will be decoded in the absence of interference from $o$ and $p$.  This increases the probability of their successful decoding.
    \end{itemize}
\end{itemize}
Based on the analysis of these three possibilities, the outcome of the decoding orders differs only in  \textit{possibility 3}. In this case, \textit{decoding order 2} clearly leads to better results. Therefore, for this given problem, the optimal decoding order is to decode the aircraft with the highest \ac{sinr} at each iteration.
\end{proof}
By Lemma~\ref{lemma:proof} we can claim that \mbox{V-BLAST}  finds an optimal solution for the problem given in \eqref{eq:probSSA}, for an equal-rate system, but it may not find the optimal solution for the variable-rate system.

\subsection{Channel Gain and Transmission Rate (CGTR)}

In \cite{order2},  the decreasing order of $|\textbf{h}_{:k}|^{2}(1+\frac{1}{2^{r_k}+1})$ is proposed as the optimal decoding order to maximize the sum rate, i.e., $\max_{\Pi}\sum_{k\in S^*_{\text{SIC}}}r_k$. For an equal-rate system, this decoding order leads to decreasing order  of $|\textbf{h}_{:k}|^{2}$ since $(1+\frac{1}{2^{r_G}+1})$ becomes a constant. Moreover, in an equal-rate system, optimizing the sum rate minimizes the outage probability, which is the objective we pursue in this paper.

Despite the fact that the study in \cite{order2} does not explicitly specify the number of antenna elements at the receiver, their approach 
applies only to single antenna receivers. 
This is because the linear dependencies between the users' channel vectors are not considered in \cite{order2}.
In the multiple antenna \ac{noma} concept, the interference caused by one user to another depends not only on the power of the interfering signal but also on the linear dependence of the users' channel vectors. If the channel vectors are orthogonal, there will be no interference. However, if the channel vectors are linearly dependent, then there can be significant interference.

%% file: Sections/07.tex
\section{Numerical Results and Discussion} \label{sec:results}
In this section, we evaluate the performance of the proposed algorithms, \ac{ssa}, \ac{gsa}, and  \ac{lgsa}. 
We also compare their performance with the decoding ordering strategies from \cite{vblast} and \cite{order2}. 
Evaluations are performed for two types of systems:  equal-rate and variable-rate. In an equal-rate system, all aircraft transmit at the guaranteed rate, $r_G$, of the system, i.e.,  \mbox{$r_k=r_G, \forall k \in [K]$}, where \mbox{$[K]=\{1,2,\dots,K\}$}. This approach allows us to understand the relationship between the transmission rate and the outage probability, $P_{\text{out}}$, and to analyze the range of rates that could be supported in the given geometry-based stochastic channel model. Based on these results, we determine the transmission rate range in variable-rate simulations. Accordingly, in the variable-rate simulations, we focus on the relationship between the number of aircraft sharing a non-orthogonal channel, $K$, and $P_{\text{out}}$, \rewone{as well as the impact of the number of antenna elements in the array on $P_{\text{out}}$.}
Finally, we evaluate the computational complexity of the proposed algorithms and compare them to a brute-force approach. 
\begin{figure}[t]
     \centering
     \resizebox{0.45\textwidth}{!}{\input{Plots/ER_SSA_GSA}}
     \caption{Outage probability, $P_{\text{out}}$, for an equal-rate system. The results are obtained for a \ac{rec} with $M=64$ and for  varying number of aircraft, $K$. The dashed line represents \ac{ssa} and the solid line corresponds to \ac{gsa}.}
     \label{fig:03_P_out}
\end{figure}
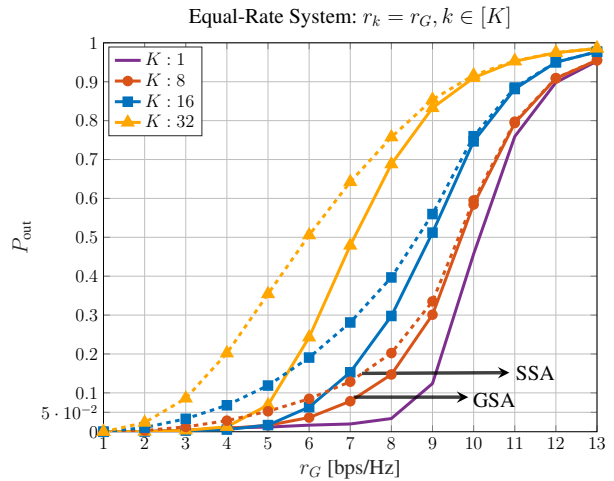

\subsection{Evaluation of Proposed Algorithms: SSA, GSA and LGSA}

Figure~\ref{fig:03_P_out} compares the achievable outage probabilities using \ac{ssa} and \ac{gsa} for an equal-rate system.
\rewfour{ In this setting, the case of $K=1$ corresponds to an \ac{oma} configuration, since only a single aircraft transmits over a given time-frequency resource block. In other words, it corresponds to a single-user multiple antenna system.}
As the \mbox{aircraft number, $K$,} increases, the difference in performance between \ac{ssa} and \ac{gsa} gets larger. This means that joint group decoding of the aircraft becomes more and more necessary to maintain a low $P_{\text{out}}$ at larger values of $K$. 

\rewthree{Focusing on the results of the \ac{gsa} in Fig.~\ref{fig:03_P_out}, for lower values of $K$, the system can achieve higher  $r_G$ values while maintaining a low $P_{\text{out}}$. In contrast, for larger values of $K$, $P_{\text{out}}$ increases rapidly as $r_G$ grows.
For example, at \mbox{$P_{\text{out}}=0.05$}, while $K=1$  can transmit at approximately $r_G=$\qty{8.2}{bps/Hz}, $K=32$ using the \ac{gsa} algorithm achieves only about $r_G=$\qty{4.8}{bps/Hz}. However, when we consider the system’s sum rate capacity, which is determined by multiplying $K$ and $r_G$  at  a fixed $P_{\text{out}}$, larger values of $K$ yield better results. In the example where \mbox{$P_{\text{out}}=0.05$}, the sum rate for $K=1$ is approximately \qty{8.2}{bps/Hz}, whereas for $K=32$, the sum rate is approximately \mbox{$32\cdot\qty{4.8}{bps/Hz}=\qty{153.6}{bps/Hz}$}. As a result, while the lower values of $K$ allow for a higher $r_G$ at a fixed $P_{\text{out}}$, the sum rate increases at larger values of $K$.}

\rewthree{The highlight of Fig.~\ref{fig:03_P_out} is that at \mbox{$r_G=2$} and \mbox{$r_G=3$}, both $K=1$ and $K=32$ achieve nearly identical  $P_{\text{out}}$ using \ac{gsa}. This means that, in theory, the same resources can be used efficiently by 32 aircraft instead of just 1, with almost no loss in performance. 
 In other words, at these transmission rates, the system's sum rate can scale effectively with $K$.
 }
 
 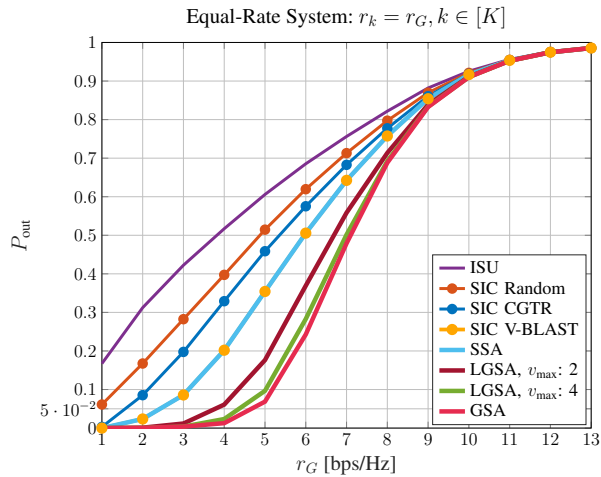
\begin{figure}[t]
    \centering
   \resizebox{0.45\textwidth}{!}{\input{Plots/ER_All}}
    \caption{Outage probability, $P_{\text{out}}$ vs. guaranteed rate $r_G$ for $K = 32$ in an equal-rate system. The results are obtained for a \ac{rec} with $M=64$.}
    \label{fig:ER_vMax}
\end{figure}

The evaluation of \ac{lgsa} is only meaningful if there is a significant difference between \ac{ssa} and \ac{gsa}, as this indicates that \ac{gsa} achieves low $P_{\text{out}}$ values by joint group decoding of aircraft. At this point, \ac{lgsa} becomes relevant, as it ensures that the number of aircraft decoded together remains below the limit, $v_{\text{max}}$. This allows \ac{lgsa} to offer a joint decoding approach that is more practical for real-world applications. 
In Fig.~\ref{fig:ER_vMax}, $P_{\text{out}}$ is computed for $K = 32$ in an equal-rate system. \ac{lgsa} is evaluated for two group size limits, \mbox{$v_{\text{max}} = 2$} and \mbox{$v_{\text{max}} = 4$}. 
\rewthree{ For both \mbox{$v_{\text{max}}$} values, \ac{lgsa} significantly outperforms \ac{ssa}.
Specifically, for \mbox{$v_{\text{max}} = 4$}, \ac{lgsa} and \ac{gsa} exhibit nearly identical outage probabilities for \mbox{$r_G \leq 2$~[bps/Hz]}.
}


We define single-user single antenna systems as those where both the \ac{gs} and the aircraft are equipped with a single antenna, and the aircraft are separated in time or frequency, i.e., \ac{oma}. 
\rewfour{In this context, both current \ac{atm} communications and the \ac{sesar} system \cite{LDACSSpec19} can be regarded as single-user single antenna systems.}
For a single-user single antenna system, 
the $P_{\text{out}}$ is 0.03 at $r_G=$\qty{2}{bps/Hz} and  0.06 at  $r_G=$ \qty{3}{bps/Hz} with a steep increase in $P_{\text{out}}$  observed for larger values of $r_G$. 
\rewfour{ When these results are compared with the proposed \ac{noma} scheme with multiple antennas, it can be observed that for \mbox{$r_G\leq$ \qty{3}{bps/Hz}}  the outage probability of individual aircraft is reduced for \ac{gsa} and \ac{lgsa} across all evaluated values of $K$ and $v_{\text{max}}$ in Fig.~\ref{fig:03_P_out} and Fig.~\ref{fig:ER_vMax}, and for \ac{ssa} in all cases except $K=32$ when $r_G=\qty{3}{bps/Hz}$. This indicates that, while the proposed system achieves higher spectral efficiency through increased sum rate capacity, it also improves the reliability of \ac{atm} communications.}

\begin{figure}[t]
    \centering
    \resizebox{0.45\textwidth}{!}{\input{Plots/VR_All}}
    \caption{Outage probability, $P_{\text{out}}$, for a variable-rate system where $r_k$ is uniformly distributed between \qty{2}{} and \qty{6}{bps/Hz}. The results are obtained for a \ac{rec} with $M=64$.}
    \label{fig:VR_vMax}
\end{figure}
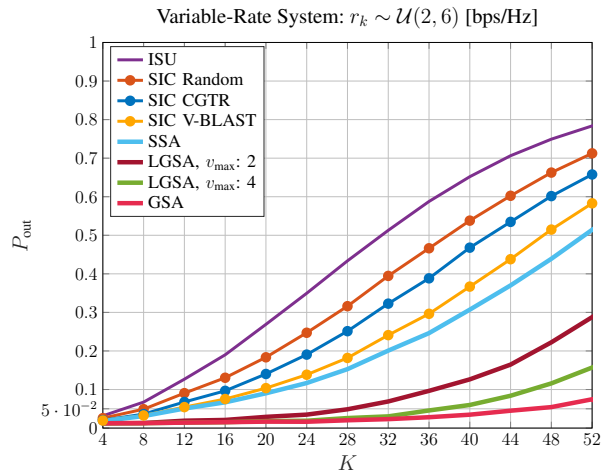

\rewfour{From a practical perspective,} it is not realistic to expect that all current single-user single antenna  \ac{gs}s worldwide will be replaced by the new multiple antenna \ac{gs}s simultaneously.
\rewfour{In such a transitional deployment scenario, }   $r_G$  will be defined based on the single-user single antenna \ac{gs}s to ensure that flight critical messages can also be supported in these stations. 
\rewfour{At the same time,} if all aircraft would transmit at the guaranteed rate of a single-user single antenna system, i.e., approximately \qty{2}{bps/Hz}, the full potential of the installed multiple antenna systems would not be utilized, thereby wasting valuable resources.
As shown in Fig.~\ref{fig:03_P_out} and Fig.~\ref{fig:ER_vMax}, the proposed algorithms, \ac{ssa}, \ac{gsa}, and \ac{lgsa}, are capable of achieving higher transmission rates while maintaining a low outage probability.  This motivates us to consider a variable-rate system where aircraft transmit at least at the guaranteed rate, but could transmit at higher rates to \rewfour{fully leverage the benefits of multiple antenna systems.}

In the variable-rate system, we set $r_G$ to  \qty{2}{bps/Hz}, which is based on the single-user single antenna system where $P_{\text{out}}$ is $0.03$. Accordingly, the transmission rate of aircraft $k$ is randomly selected from a uniform distribution between  \qty{2}{} and  \qty{6}{bps/Hz}, i.e., \mbox{$r_k\sim \mathcal{U}(2,6), \forall k \in [K]$}.
In Fig.~\ref{fig:VR_vMax}, $P_{\text{out}}$ is computed for increasing values of $K$ in a variable-rate system. 
It can be observed that \ac{lgsa} with \mbox{$v_{\text{max}} = 4$} has the same performance as \ac{gsa}  up to $K = 12$. Moreover, for both $v_{\text{max}}=2$ and  $v_{\text{max}}=4$, \ac{lgsa} outperforms \ac{ssa} by a clear margin.
\rewone{Moreover, in Fig.~\ref{fig:VR_K32}, $P_{\text{out}}$ is computed for increasing values of  $M$ in a variable-rate system with $K=32$. For $M<36$, the $P_{\text{out}}$ is very high for all three algorithms: \ac{ssa}, \ac{lgsa}, and \ac{gsa}. As $M$ increases,  $P_{\text{out}}$ decreases and the performance gap among the algorithms narrows. Specifically, at $M=144$, \ac{lgsa} with $v_{\text{max}}=4$ and  \ac{gsa} achieve nearly identical  $P_{\text{out}}$.}

The results in Fig.~\ref{fig:VR_vMax} and  Fig.~\ref{fig:VR_K32} show that even when joint group decoding is used for groups of only 2, the outage probability can still be significantly reduced compared to \ac{ssa}, especially for large values of $K$. 
In this paper, aircraft are randomly assigned transmission rates during the variable-rate simulations. However, the performance of the system can be improved by determining transmission rates based on factors such as the power of the received signal or the distance between the aircraft and the \ac{gs}.

\subsection{Performance Comparison: Proposed Algorithms vs. Existing Decoding Strategies}
The outage probability of the \ac{isu} decoders and of the \ac{sic} decoding order strategies proposed in \cite{vblast} and \cite{order2} are also plotted in Fig.~\ref{fig:ER_vMax}, Fig.~\ref{fig:VR_vMax}, and Fig.~\ref{fig:VR_K32}.
In these figures, the minimal outage probability for \ac{isu} decoders is notably very high. To achieve a low $P_{\text{out}}$ with \ac{isu} decoders, the number of simultaneously transmitting aircraft, $K$, must be very small compared to $M$. Comparing the \ac{isu} decoders with the \ac{sic} decoders, we observe that the \ac{sic} scheme allows more aircraft to transmit over a non-orthogonal channel with lower outage probabilities. As a result, these findings demonstrate how the \ac{sic} scheme reduces the outage probability while also increasing the spectral efficiency when compared to the \ac{isu} decoders.

\begin{figure}[t]
    \centering
    \resizebox{0.45\textwidth}{!}{\input{Plots/VR_M_K32}}
    \caption{\rewone{Outage probability, $P_{\text{out}}$, for a variable-rate system where $r_k$ is uniformly distributed between \qty{2}{} and \qty{6}{bps/Hz}. The results are obtained for a \ac{rec} with a varying number of antenna elements and  $K=32$.}}
    \label{fig:VR_K32}
\end{figure}

\ac{ssa} obtains the optimal decoding order for the \ac{sic} ordering problem given in \eqref{eq:probSSA}. In order to demonstrate the importance of the decoding order, we first present the results for a random decoding order in Fig.~\ref{fig:ER_vMax}, Fig.~\ref{fig:VR_vMax}, and Fig.~\ref{fig:VR_K32}, denoted as \ac{sic} random. Not surprisingly, \ac{sic} random shows the worst performance, with a clear margin between \ac{sic} random and \ac{ssa} in all three figures. 
The decoding order strategy proposed in \cite{order2}, \ac{sic} \ac{cgtr}, has the worst performance after \ac{sic} random across all figures. The authors of \cite{order2} argued that they propose an optimal decoding order for a power-controlled system. 
Although not explicitly stated in \cite{order2}, their approach is only applicable to single antenna receivers and not to multiple antenna receivers. This is because the linear dependencies between the users' (aircraft's) channel vectors are ignored when calculating the achievable rates.
Lastly, we evaluate the well-known V-BLAST~\cite{vblast}.  As proved in Lemma~\ref{lemma:proof} and shown in Fig.~\ref{fig:ER_vMax}, when an equal-rate system is considered, V-BLAST obtains an optimal decoding order. 
For all channel realizations, the set of successfully decoded aircraft by V-BLAST is the same as that of \ac{ssa} in an equal-rate system. However, in the variable-rate system, as can be seen in Fig.~\ref{fig:VR_vMax} and Fig.~\ref{fig:VR_K32}, \ac{ssa} outperforms V-BLAST.  

\begin{figure}[t]
    \centering
    \resizebox{0.45\textwidth}{!}{\input{Plots/ER_CompCount}}
    \caption{The average complex multiplications required to compute outage probability, $P_{\text{out}}$, for an equal-rate system. The dashed line represents \ac{ssa} and the solid line is used for \ac{lgsa} and \ac{gsa}. The results are obtained for a \ac{rec} with $M=64$ and $q_{\text{max}}=2$.}
    \label{fig:ER_compCount}
\end{figure}
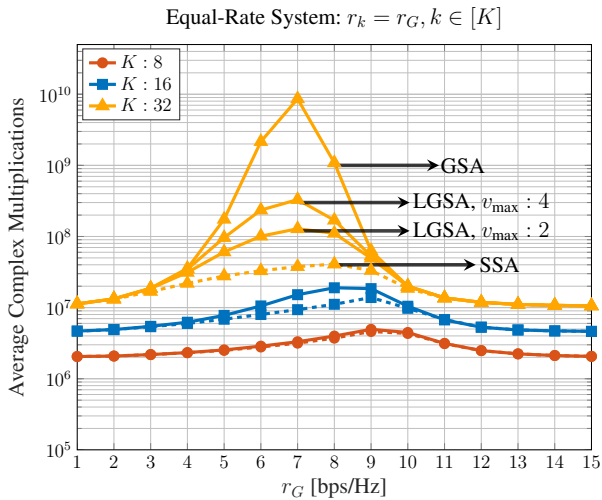

\subsection{Computational Complexity of the Proposed Algorithms}
Next, we evaluate the computational complexity of the proposed algorithms based on the average number of complex multiplications required for a given channel matrix and rate point. This is computed based on \cite{matrix_book}, where the multiplication of a $M\times K$ matrix with a $K\times M$ matrix costs $M^2K$ multiplications and the inverse of an $M\times M$ matrix costs $M^3$ multiplications.

Figure~\ref{fig:ER_compCount} shows the results for an equal-rate system and Fig.~\ref{fig:VR_compCount} presents the results for a variable-rate system.
In both cases the computational complexity increases with $K$. 
\rewfour{In Fig.~\ref{fig:ER_compCount}, it can be observed that the computational complexity of the \ac{gsa} varies significantly with $r_G$.
In particular,} there is a notable peak at $r_G=\qty{7}{bps/Hz}$ for \ac{gsa} when $K=32$.
In this specific case, many users can be decoded under the interference of the outage set. 
However, only a few aircraft can achieve \qty{7}{bps/Hz} under the interference of their constraint set. This results in many aircraft remaining in $L$ during the \textit{GreedyGroup} phase of Algorithm~\ref{al:step2} (\ac{gsa}), which significantly increases the computational load. Nevertheless, for $r_G=\qty{7}{bps/Hz}$ with \ac{gsa} and $K=32$, the equal-rate system ends up with a $P_{\text{out}}$ of about $0.5$  (see Fig.~\ref{fig:ER_vMax}). 
The computational complexity for  $r_G$ values that achieve a  $P_{\text{out}}$ of less than $0.01$ is particularly relevant to our analysis. In these cases, the average number of complex multiplications per channel matrix remains below $1\cdot10^{8}$ in Fig.~\ref{fig:ER_compCount}, \rewfour{ and the computational complexity of the \ac{gsa} approaches that of the \ac{ssa}.}
\rewfour{A similar trend is observed in Fig.~\ref{fig:VR_compCount}, where the computational complexity of the \ac{gsa} closely matches that of the \ac{ssa} for $P_{\text{out}}\leq 0.01$, corresponding to  $K\leq 32$.
 Therefore, it can be argued that  the computational complexity of the \ac{gsa} does not differ significantly from that of the \ac{ssa} within the operating ranges where the \ac{gsa} achieves low outage probabilities. Nevertheless, the performance advantage of the \ac{gsa} over the \ac{ssa} in these regions remains substantial.}
Moreover, in Fig.~\ref{fig:VR_compCount}, the average number of complex multiplications remains below $1\cdot10^{8}$ for $K\leq 36$.

If we had to compute the number of aircraft in outage using a brute-force approach for $K=32$, according to \eqref{eq:comp_out} we would have to evaluate \eqref{eq:con_GSA} $1.8\cdot10^{15}$ times for a given channel matrix and a rate point. For $M=64$, this would cost about $1.2\cdot10^{21}$ complex multiplications. For example, suppose  we use a very powerful processor with a clock rate of \qty{9}{GHz} and each complex multiplication takes one clock cycle.  Then it would take us about 3000 years to compute the number of aircraft in outage  for a given channel matrix and a rate point with a brute-force approach. 
By comparison, the calculation of the $1\cdot10^{8}$ complex multiplications on the same processor would take about \qty{0.01}{s}.

%% file: Plots/ER_SSA_GSA.tex
%
%
\definecolor{mycolor1}{rgb}{0.85098,0.32549,0.09804}
\definecolor{mycolor2}{rgb}{0.00000,0.44706,0.74118}
\definecolor{mycolor3}{rgb}{1.0, 0.65, 0.0}
\definecolor{mycolor4}{rgb}{0.49412,0.18431,0.55686}
\pgfplotsset{compat=1.10}
\begin{tikzpicture}

\begin{axis}[%
width=4.521in,
height=3.566in,
at={(0.758in,0.481in)},
scale only axis,
title  ={\Large Equal-Rate System: $r_k=r_G, k\in [K]$},
xmin=2,
xmax=10,
xtick={2, 2, 3, ..., 15},
xticklabel style = {font=\large},
xlabel style={font=\color{white!15!black}},
xlabel={\Large $r_G$ [bps/Hz]},
ymode=log,
ymin=3e-4,
ymax=1,
ylabel style={font=\color{white!15!black}},
yticklabel style = {font=\large},
ylabel={\Large $P_{\text{out}}$},
axis background/.style={fill=white},
extra y ticks={0.05},
extra y tick style={xticklabel=\pgfmathprintnumber{\tick}},
xmajorgrids,
ymajorgrids,
yminorgrids,
legend style={at={(0.99,0.25)}, legend cell align=left, align=left, draw=white!15!black}
]

\addplot [color=mycolor4, line width=2.0pt, mark options={solid, mycolor4}]
  table[row sep=crcr]{%
1	0\\
2	0.000429687500000053\\
3	0.002698125\\
4	0.00563437499999997\\
5	0.00992656250000001\\
6	0.0161371875\\
7	0.0256046875\\
8	0.0422740625\\
9	0.1509109375\\
10	0.493485\\
11	0.7410790625\\
12	0.8687203125\\
13	0.953\\
14	0.977\\
15	0.987\\
};
\addlegendentry{\large $K:1$ \rewfour{(OMA)}}

\addplot [color=mycolor1, line width=2.0pt, mark=*, mark size=2.5000pt, mark options={solid, mycolor1}] table[row sep=crcr]{%
1	0\\
2	0.000433437500000022\\
3	0.00271687499999995\\
4	0.00575375\\
5	0.0112978124999999\\
6	0.0306465625\\
7	0.0680403125\\
8	0.1365353125\\
9	0.3193503125\\
10	0.6284265625\\
11	0.79375\\
12	0.9085\\
13	0.955125\\
14	0.9775\\
15	0.989375\\
16	0.995875\\
17	0.998125\\
18	0.999125\\
19	1\\
20	1\\
};
\addlegendentry{\large $K:8$}

\addplot [color=mycolor2, line width=2.0pt, mark=square*, mark size=2.5000pt, mark options={solid, mycolor2}]
  table[row sep=crcr]{%
1	0\\
2	0.000434687499999975\\
3	0.00272656250000003\\
4	0.00601968750000004\\
5	0.0160253125000001\\
6	0.0629790625\\
7	0.1555296875\\
8	0.300100625\\
9	0.5333690625\\
10	0.7675459375\\
11	0.8811875\\
12	0.95\\
13	0.9769375\\
14	0.9863125\\
15	0.9925\\
16	0.9964375\\
17	0.998625\\
18	1\\
19	1\\
20	1\\
};
\addlegendentry{\large $K:16$}

\addplot [color=mycolor3, line width=2.0pt, mark=triangle*, mark size=3.5000pt, mark options={solid, mycolor3}]
  table[row sep=crcr]{%
1	0\\
2	0.000441562500000048\\
3	0.00276343749999997\\
4	0.00879593749999996\\
5	0.0596840625\\
6	0.24021724455201\\
7	0.478999498495487\\
8	0.688063063063063\\
9	0.83259375\\
10	0.91021875\\
11	0.9524375\\
12	0.97453125\\
13	0.9855625\\
14	0.99115625\\
15	0.99459375\\
16	0.99721875\\
17	0.9986875\\
18	0.99953125\\
19	1\\
20	1\\
};
\addlegendentry{\large  $K:32$}

\addplot  [color=mycolor1, line width=2.0pt, dashed, mark=*, mark size=2.5000pt, mark options={solid, mycolor1}] table[row sep=crcr]{%
1	2.81250000000455e-06\\
2	0.004953125\\
3	0.0148871875\\
4	0.02929875\\
5	0.05038875\\
6	0.0808653125\\
7	0.125236875\\
8	0.1922521875\\
9	0.3510121875\\
10	0.6378171875\\
11	0.79725\\
12	0.90925\\
13	0.955125\\
14	0.9775\\
15	0.989375\\
16	0.995875\\
17	0.998125\\
18	0.999125\\
19	1\\
20	1\\
};

\addplot [color=mycolor2, line width=2.0pt, dashed, mark=square*,, mark size=2.5000pt, mark options={solid, mycolor2}]
  table[row sep=crcr]{%
1	8.12500000002547e-06\\
2	0.010431875\\
3	0.03220375\\
4	0.0674353125\\
5	0.1193709375\\
6	0.1904184375\\
7	0.2820734375\\
8	0.3986828125\\
9	0.5770959375\\
10	0.7780565625\\
11	0.885\\
12	0.9503125\\
13	0.9769375\\
14	0.9863125\\
15	0.9925\\
16	0.9964375\\
17	0.998625\\
18	1\\
19	1\\
20	1\\
};

\addplot [color=mycolor3, line width=2.0pt, dashed, mark=triangle*, mark size=3.5000pt, mark options={solid, mycolor3}]
  table[row sep=crcr]{%
1	5.43750000000509e-05\\
2	0.025281875\\
3	0.09543625\\
4	0.2167878125\\
5	0.3673934375\\
6	0.5206396875\\
7	0.642375\\
8	0.75746875\\
9	0.85409375\\
10	0.91671875\\
11	0.95365625\\
12	0.97475\\
13	0.9855625\\
14	0.99115625\\
15	0.99459375\\
16	0.99721875\\
17	0.9986875\\
18	0.99953125\\
19	1\\
20	1\\
};

\node(S_origin) at (axis cs:5.5, 0.075){};
\node(S_destination) at (axis cs:5.5,	0.006){};
\node(G_origin) at (axis cs:7.5, 0.11){};
\node(G_destination) at (axis cs:7.5,	0.006){};

\draw[-stealth,  line width=0.8mm, draw opacity=0.751] (S_origin)--(S_destination);
\draw[-stealth,  line width=0.8mm, draw opacity=0.751] (G_origin)--(G_destination);

\node[black] (SSA) at (axis cs:5.5, 0.005){\Large SSA};
\node[black](GSA) at (axis cs:7.5,	0.005){\Large GSA};

\end{axis}
\end{tikzpicture}%

%% file: Plots/ER_All.tex
%
%

\definecolor{mycolor1}{rgb}{0.49412,0.18431,0.55686}
\definecolor{mycolor2}{rgb}{0.85098,0.32549,0.09804}
\definecolor{mycolor3}{rgb}{0.00000,0.44706,0.74118}
\definecolor{mycolor4}{rgb}{1.0, 0.65, 0.0}
\definecolor{mycolor5}{rgb}{0.30196,0.74510,0.93333}
\definecolor{mycolor6}{rgb}{0.63529,0.07843,0.18431}
\definecolor{mycolor7}{rgb}{0.46667,0.67451,0.18824}
\definecolor{mycolor8}{rgb}{0.9, 0.17, 0.31} 
\definecolor{mycolor9}{rgb}{0.0, 0.0, 1.0} 
\definecolor{mycolor10}{rgb}{0.5, 0.5, 0.0} 

\pgfplotsset{compat=1.10}
\begin{tikzpicture}

\begin{axis}[%
width=4.521in,
height=3.566in,
at={(0.758in,0.481in)},
scale only axis,
title  ={\Large Equal-Rate System: $r_k=r_G, k\in [K]$},
xmin=2,
xmax=10,
xtick={1, 2, 3, ..., 15},
xticklabel style = {font=\large},
xlabel style={font=\color{white!15!black}},
xlabel={\Large $r_G$ [bps/Hz]},
ymode=log,
ymin=3e-4,
ymax=1,
ylabel style={font=\color{white!15!black}},
yticklabel style = {font=\large},
ylabel={\Large $P_{\text{out}}$},
extra y ticks={0.05},
extra y tick style={xticklabel=\pgfmathprintnumber{\tick}},
axis background/.style={fill=white},
xmajorgrids,
ymajorgrids,
yminorgrids,
legend style={at={(0.99,0.45)}, legend cell align=left, align=left, draw=white!15!black}
]

\addplot [color=mycolor1, line width=3.0pt, mark options={solid, mycolor1}]
  table[row sep=crcr]{%
1	0.16102375\\
2	0.31829125\\
3	0.435076875\\
4	0.5315103125\\
5	0.6169525\\
6	0.6941459375\\
7	0.75615625\\
8	0.82178125\\
9	0.88190625\\
10	0.92534375\\
11	0.956125\\
12	0.9754375\\
13	0.98575\\
14	0.9911875\\
15	0.99459375\\
};
\addlegendentry{\large ISU}

\addplot [color=mycolor2, line width=2.0pt, mark=*, mark size=2.5000pt, mark options={solid, mycolor2}]
  table[row sep=crcr]{%
1	0.0578265625\\
2	0.17123625\\
3	0.2919803125\\
4	0.4120809375\\
5	0.526210625\\
6	0.630324375\\
7	0.71296875\\
8	0.7971875\\
9	0.86965625\\
10	0.92115625\\
11	0.95478125\\
12	0.975125\\
13	0.98565625\\
14	0.99115625\\
15	0.99459375\\
};
\addlegendentry{\large SIC Random}

\addplot [color=mycolor3, line width=2.0pt, mark=*, mark size=2.5000pt, mark options={solid, mycolor3}]
  table[row sep=crcr]{%
1	0.00472031250000005\\
2	0.09708125\\
3	0.2203134375\\
4	0.351218125\\
5	0.4786525\\
6	0.595155\\
7	0.68265625\\
8	0.7774375\\
9	0.86075\\
10	0.9183125\\
11	0.9541875\\
12	0.97496875\\
13	0.98559375\\
14	0.99115625\\
15	0.99459375\\
};
\addlegendentry{\large SIC CGTR}


\addplot [color=mycolor4, line width=2.5pt, mark=*, mark size=2.5000pt, mark options={solid, mycolor4}]
  table[row sep=crcr]{%
1	5.43750000000509e-05\\
2	0.025281875\\
3	0.09543625\\
4	0.2167878125\\
5	0.3673934375\\
6	0.5206396875\\
7	0.642375\\
8	0.75746875\\
9	0.85409375\\
10	0.91671875\\
11	0.95365625\\
12	0.97475\\
13	0.9855625\\
14	0.99115625\\
15	0.99459375\\
};
\addlegendentry{\large SIC V-BLAST}
\addplot [color=mycolor5, line width=3.0pt,  mark options={solid, mycolor5}]
  table[row sep=crcr]{%
1	5.43750000000509e-05\\
2	0.025281875\\
3	0.09543625\\
4	0.2167878125\\
5	0.3673934375\\
6	0.5206396875\\
7	0.642375\\
8	0.75746875\\
9	0.85409375\\
10	0.91671875\\
11	0.95365625\\
12	0.97475\\
13	0.9855625\\
14	0.99115625\\
15	0.99459375\\
16	0.99721875\\
17	0.9986875\\
18	0.99953125\\
19	1\\
20	1\\
};
\addlegendentry{\large SSA}

\addplot [color=mycolor6, line width=3.0pt, mark options={solid, mycolor6}]
  table[row sep=crcr]{%
1	6.87499999996177e-06\\
2	0.00138156249999999\\
3	0.0130065625\\
4	0.062444375\\
5	0.1843134375\\
6	0.37744375\\
7	0.570156875\\
8	0.71346875\\
9	0.83884375\\
10	0.9113125\\
11	0.9525625\\
12	0.97453125\\
13	0.9855625\\
14	0.99115625\\
15	0.99459375\\
};
\addlegendentry{\large LGSA, $v_{\text{max}}\text{: 2}$}

\addplot [color=mycolor7, line width=3.0pt, mark options={solid, mycolor7}]
  table[row sep=crcr]{%
1	0\\
2	0.000445312500000017\\
3	0.00363093749999999\\
4	0.0202393749999999\\
5	0.09562875\\
6	0.2823778125\\
7	0.5015625\\
8	0.6923125\\
9	0.8336875\\
10	0.91021875\\
11	0.9524375\\
12	0.97453125\\
13	0.9855625\\
14	0.99115625\\
15	0.99459375\\
};
\addlegendentry{\large LGSA, $v_{\text{max}}\text{: 4}$}

\addplot [color=mycolor8, line width=3.0pt,  mark options={solid, mycolor8}] table[row sep=crcr]{%
1	0\\
2	0.000441562500000048\\
3	0.00276343749999997\\
4	0.00879593749999996\\
5	0.0596840625\\
6	0.24021724455201\\
7	0.478999498495487\\
8	0.688063063063063\\
9	0.83259375\\
10	0.91021875\\
11	0.9524375\\
12	0.97453125\\
13	0.9855625\\
14	0.99115625\\
15	0.99459375\\
16	0.99721875\\
17	0.9986875\\
18	0.99953125\\
19	1\\
20	1\\
};
\addlegendentry{\large GSA}

\end{axis}
\end{tikzpicture}%

%% file: Plots/VR_All.tex
%
%

\definecolor{mycolor1}{rgb}{0.49412,0.18431,0.55686}
\definecolor{mycolor2}{rgb}{0.85098,0.32549,0.09804}
\definecolor{mycolor3}{rgb}{0.00000,0.44706,0.74118}
\definecolor{mycolor4}{rgb}{1.0, 0.65, 0.0}
\definecolor{mycolor5}{rgb}{0.30196,0.74510,0.93333}
\definecolor{mycolor6}{rgb}{0.63529,0.07843,0.18431}
\definecolor{mycolor7}{rgb}{0.46667,0.67451,0.18824}
\definecolor{mycolor8}{rgb}{0.9, 0.17, 0.31} 
\definecolor{mycolor9}{rgb}{0.0, 0.0, 1.0} 

\pgfplotsset{compat=1.10}
\begin{tikzpicture}

\begin{axis}[%
width=4.521in,
height=3.566in,
at={(0.758in,0.481in)},
scale only axis,
title  ={\Large Variable-Rate System: $r_k\sim\mathcal{U}(2, 6)$  [bps/Hz]},
xmin=4,
xmax=52,
xtick={4,8,12,...,64},
xticklabel style = {font=\large},
xlabel style={font=\color{white!15!black}},
xlabel={\Large $K$},
ymode=log,
ymin=1e-3,
ymax=1,
ylabel style={font=\color{white!15!black}},
yticklabel style = {font=\large},
ylabel={\Large $P_{\text{out}}$},
extra y ticks={0.05},
extra y tick style={xticklabel=\pgfmathprintnumber{\tick}},
axis background/.style={fill=white},
xmajorgrids,
ymajorgrids,
yminorgrids,
legend style={at={(0.99,0.45)}, legend cell align=left, align=left, draw=white!15!black}
]

\addplot [color=mycolor1, line width=3.0pt,  mark options={solid, mycolor1}]
  table[row sep=crcr]{%
4	0.0317499999999999\\
8	0.067375\\
12	0.12675\\
16	0.19\\
20	0.269\\
24	0.349625\\
28	0.433678571428571\\
32	0.51278125\\
36	0.587694444444445\\
40	0.6519\\
44	0.706295454545455\\
48	0.749208333333333\\
52	0.783596153846154\\
56	0.817660714285714\\
60	0\\
64	0\\
};
\addlegendentry{\large ISU}

\addplot [color=mycolor2, line width=2.0pt, mark=*, mark size=2.5000pt, mark options={solid, mycolor2}]
  table[row sep=crcr]{%
4	0.0257500000000001\\
8	0.049\\
12	0.0905833333333333\\
16	0.1305\\
20	0.18355\\
24	0.247208333333333\\
28	0.316107142857143\\
32	0.394625\\
36	0.466277777777778\\
40	0.537975\\
44	0.602386363636364\\
48	0.6625625\\
52	0.712423076923077\\
56	0.756946428571429\\
60	0\\
64	0\\
};
\addlegendentry{\large SIC Random}

\addplot [color=mycolor3, line width=2.0pt, mark=*, mark size=2.5000pt, mark options={solid, mycolor3}]
  table[row sep=crcr]{%
4	0.02025\\
8	0.03625\\
12	0.0676666666666667\\
16	0.096625\\
20	0.1403\\
24	0.190541666666667\\
28	0.251321428571429\\
32	0.32271875\\
36	0.388333333333333\\
40	0.467875\\
44	0.534727272727273\\
48	0.6015625\\
52	0.657634615384615\\
56	0.714035714285714\\
60	0.755733333333333\\
64	0.7898125\\
};
\addlegendentry{\large SIC CGTR}

\addplot [color=mycolor4, line width=2.0pt, mark=*, mark size=2.5000pt, mark options={solid, mycolor4}]
  table[row sep=crcr]{%
4	0.0195\\
8	0.0327499999999999\\
12	0.0544166666666667\\
16	0.075375\\
20	0.1036\\
24	0.138416666666667\\
28	0.181892857142857\\
32	0.24103125\\
36	0.296361111111111\\
40	0.366875\\
44	0.437931818181818\\
48	0.51475\\
52	0.582826923076923\\
56	0.653089285714286\\
60	0\\
64	0\\
};
\addlegendentry{\large SIC V-BLAST}
\addplot [color=mycolor5, line width=3.0pt,  mark size=2.5000pt, mark options={solid, mycolor5}]
  table[row sep=crcr]{%
4	0.01825\\
8	0.031125\\
12	0.0506666666666667\\
16	0.0668125000000001\\
20	0.0901999999999999\\
24	0.116833333333333\\
28	0.152857142857143\\
32	0.200875\\
36	0.246277777777778\\
40	0.3071\\
44	0.37\\
48	0.438979166666667\\
52	0.514134615384615\\
56	0.580928571428571\\
60	0\\
64	0\\
};
\addlegendentry{\large  SSA}

\addplot [color=mycolor6, line width=3.0pt, mark options={solid, mycolor6}]
  table[row sep=crcr]{%
4	0.00814545454545457\\
8	0.00963000000000003\\
12	0.0120611111111111\\
16	0.0171\\
20	0.0237624999999999\\
24	0.0353916666666666\\
28	0.0484607142857143\\
32	0.07040625\\
36	0.0958666666666667\\
40	0.127285\\
44	0.172379545454545\\
48	0.222375\\
52	0.287711538461538\\
56	0.3500\\
60	0.4265\\
64	0.5204\\
};
\addlegendentry{\large LGSA, $v_{\text{max}}\text{: 2}$}

\addplot [color=mycolor7, line width=3.0pt, mark options={solid, mycolor7}]
  table[row sep=crcr]{%
4	0.00797499999999995\\
8	0.00848749999999998\\
12	0.0089111111111112\\
16	0.010309375\\
20	0.0122199999999999\\
24	0.0156083333333333\\
28	0.0211178571428572\\
32	0.030121875\\
36	0.0419416666666668\\
40	0.0586675\\
44	0.0826454545454546\\
48	0.115875\\
52	0.156942307692308\\
56	0.2058\\
60	0.2650\\
64	0.3353\\
};
\addlegendentry{\large LGSA, $v_{\text{max}}\text{: 4}$}

\addplot [color=mycolor8, line width=3.0pt,  mark options={solid, mycolor8}] table[row sep=crcr]{%
4	0.00797499999999995\\
8	0.00846000000000002\\
12	0.00884999999999991\\
16	0.00969687500000005\\
20	0.0108824999999999\\
24	0.0124708333333333\\
28	0.0149285714285714\\
32	0.0187625\\
36	0.0232805555555555\\
40	0.0295925\\
44	0.0390659825514224\\
48	0.0544004676018703\\
52	0.0747175359851416\\
56	0\\
60	0\\
64	0\\
};
\addlegendentry{\large GSA}

\end{axis}
\end{tikzpicture}%

%% file: Plots/VR_M_K32.tex
%
%

\definecolor{mycolor1}{rgb}{0.49412,0.18431,0.55686}
\definecolor{mycolor2}{rgb}{0.85098,0.32549,0.09804}
\definecolor{mycolor3}{rgb}{0.00000,0.44706,0.74118}
\definecolor{mycolor4}{rgb}{1.0, 0.65, 0.0}
\definecolor{mycolor5}{rgb}{0.30196,0.74510,0.93333}
\definecolor{mycolor6}{rgb}{0.63529,0.07843,0.18431}
\definecolor{mycolor7}{rgb}{0.46667,0.67451,0.18824}
\definecolor{mycolor8}{rgb}{0.9, 0.17, 0.31} 
\definecolor{mycolor9}{rgb}{0.0, 0.0, 1.0} 
\definecolor{mycolor10}{rgb}{0.5, 0.5, 0.0} 

\pgfplotsset{compat=1.10}
\begin{tikzpicture}

\begin{axis}[%
width=4.521in,
height=3.566in,
at={(0.758in,0.481in)},
scale only axis,
title  ={\Large Variable-Rate System: $r_k\sim\mathcal{U}(2, 6)$  [bps/Hz]},
xmin=4,
xmax=144,
xtick={4,9,16,25, 36, 49, 64, 81, 100, 121,144},
xticklabel style = {font=\large},
xlabel style={font=\color{white!15!black}},
xlabel={\Large $M$},
ymode=log,
ymin=1e-3,
ymax=1,
ylabel style={font=\color{white!15!black}},
yticklabel style = {font=\large},
ylabel={\Large $P_{\text{out}}$},
extra y ticks={0.05},
extra y tick style={xticklabel=\pgfmathprintnumber{\tick}},
axis background/.style={fill=white},
xmajorgrids,
ymajorgrids,
yminorgrids,
legend style={at={(0.32,0.45)}, legend cell align=left, align=left, draw=white!15!black}
]

\addplot [color=mycolor1, line width=3.0pt,  mark options={solid, mycolor1}]
  table[row sep=crcr]{%
4	0.9951875\\
9	0.977625\\
16	0.9384375\\
25	0.86096875\\
36	0.7521875\\
49	0.6285625\\
64	0.5115625\\
81	0.41434375\\
100	0.3305\\
121	0.2681875\\
144	0.21634375\\
};
\addlegendentry{\large ISU}

\addplot [color=mycolor2, line width=2.0pt, mark=*, mark size=2.5000pt, mark options={solid, mycolor2}]
  table[row sep=crcr]{%
4	0.9950625\\
9	0.9766875\\
16	0.93053125\\
25	0.8220625\\
36	0.6728125\\
49	0.51765625\\
64	0.3885625\\
81	0.29240625\\
100	0.22659375\\
121	0.1766875\\
144	0.140625\\
};
\addlegendentry{\large SIC Random}

\addplot [color=mycolor3, line width=2.0pt, mark=*, mark size=2.5000pt, mark options={solid, mycolor3}]
  table[row sep=crcr]{%
4	0.9949375\\
9	0.9758125\\
16	0.9239375\\
25	0.79003125\\
36	0.61278125\\
49	0.44221875\\
64	0.31790625\\
81	0.231375\\
100	0.17034375\\
121	0.1276875\\
144	0.096875\\
};
\addlegendentry{\large SIC CGTR}

\addplot [color=mycolor4, line width=2.0pt, mark=*, mark size=2.5000pt, mark options={solid, mycolor4}]
  table[row sep=crcr]{%
4	0.995\\
9	0.97578125\\
16	0.920625\\
25	0.76334375\\
36	0.5390625\\
49	0.35396875\\
64	0.2325625\\
81	0.15928125\\
100	0.113\\
121	0.08140625\\
144	0.06340625\\
};
\addlegendentry{\large SIC V-BLAST}
\addplot [color=mycolor5, line width=3.0pt,  mark size=2.5000pt, mark options={solid, mycolor5}]
  table[row sep=crcr]{%
4	0.994875\\
9	0.97553125\\
16	0.916\\
25	0.72471875\\
36	0.4704375\\
49	0.3001875\\
64	0.1944375\\
81	0.13140625\\
100	0.0971562500000001\\
121	0.0701875\\
144	0.05403125\\
};
\addlegendentry{\large  SSA}

\addplot [color=mycolor6, line width=3.0pt, mark options={solid, mycolor6}]
  table[row sep=crcr]{%
4	0.994875\\
9	0.9741875\\
16	0.89171875\\
25	0.56678125\\
36	0.26824375\\
49	0.130415625\\
64	0.068865625\\
81	0.038946875\\
100	0.023775\\
121	0.015496875\\
144	0.0112125\\
};
\addlegendentry{\large LGSA, $v_{\text{max}}\text{: 2}$}

\addplot [color=mycolor7, line width=3.0pt, mark options={solid, mycolor7}]
  table[row sep=crcr]{%
4	0.994875\\
9	0.97371875\\
16	0.868875\\
25	0.42646875\\
36	0.15318125\\
49	0.064890625\\
64	0.0308875\\
81	0.017690625\\
100	0.01094375\\
121	0.00761875000000001\\
144	0.00595000000000001\\
};
\addlegendentry{\large LGSA, $v_{\text{max}}\text{: 4}$}

\addplot [color=mycolor8, line width=3.0pt,  mark options={solid, mycolor8}] table[row sep=crcr]{%
4	0.995490625\\
9   0.980359788494942 \\
16	0.837730532786885\\
25	0.293223420260782\\
36	0.0996943444344435\\
49	0.040765625\\
64	0.02018125\\
81	0.01305625\\
100	0.00930312499999997\\
121	0.00687499999999996\\
144	0.00565312500000004\\
};
\addlegendentry{\large GSA}

\end{axis}
\end{tikzpicture}%

%% file: Plots/ER_CompCount.tex
%
%
\definecolor{mycolor1}{rgb}{0.85098,0.32549,0.09804}
\definecolor{mycolor2}{rgb}{0.00000,0.44706,0.74118}
\definecolor{mycolor3}{rgb}{1.0, 0.65, 0.0}
\definecolor{mycolor4}{rgb}{0.49412,0.18431,0.55686}
\pgfplotsset{compat=1.10}
\begin{tikzpicture}

\begin{axis}[%
width=4.521in,
height=3.566in,
at={(0.758in,0.481in)},
scale only axis,
title  ={\Large Equal-Rate System: $r_k=r_G, k\in [K]$},
xmin=1,
xmax=15,
xtick={1, 2, 3, ..., 15},
xticklabel style = {font=\large},
xlabel style={font=\color{white!15!black}},
xlabel={\Large $r_G$ [bps/Hz]},
ymode=log,
yminorticks=true,
ymin=1e5,
ymax=5e10,
ylabel style={font=\color{white!15!black}},
yticklabel style = {font=\large},
ylabel={\Large Average Complex Multiplications},
axis background/.style={fill=white},
xmajorgrids,
yminorgrids,
ymajorgrids,
legend style={at={(0.2,0.99)}, legend cell align=left, align=left, draw=white!15!black}
]

\addplot [color=mycolor1, line width=2.0pt, mark=*, mark size=2.5000pt, mark options={solid, mycolor1}] table[row sep=crcr]{%
1	2051253.9136\\
2	2088049.32608\\
3	2180188.47744\\
4	2308397.24032\\
5	2499954.5344\\
6	2777764.49536\\
7	3177881.42592\\
8	3761248.6144\\
9	5037896.98048\\
10	4190679.83872\\
11	3139739.648\\
12	2487910.4\\
13	2236588.032\\
14	2121064.448\\
15	2062118.912\\
16	2029953.024\\
17	2019635.2\\
18	2014945.28\\
19	2011136\\
20	2011136\\
};
\addlegendentry{\large $K:8$}

\addplot [color=mycolor2, line width=2.0pt, mark=square*, mark size=2.5000pt, mark options={solid, mycolor2}]
  table[row sep=crcr]{%
1	4675536.54784\\
2	4912184.05376\\
3	5475256.2176\\
4	6390218.71104\\
5	8058897.3056\\
6	11065205.80096\\
7	15642087.89504\\
8	20047772.89728\\
9	19061802.20928\\
10	9897373.4912\\
11	6715797.504\\
12	5302919.168\\
13	4848414.72\\
14	4705730.56\\
15	4632113.152\\
16	4587171.84\\
17	4564615.168\\
18	4550656\\
19	4550656\\
20	4550656\\
};
\addlegendentry{\large $K:16$}

\addplot [color=mycolor3, line width=2.0pt, mark=triangle*, mark size=3.5000pt, mark options={solid, mycolor3}]
  table[row sep=crcr]{%
1	11236258.24256\\
2	13455306.17856\\
3	19809896.40704\\
4	46891116.70784\\
5	515675974.8608\\
6	2158466322.432\\
7	8622375409.10732\\
8	1086336552.48849\\
9	63149047.808\\
10	20023447.552\\
11	13697368.064\\
12	11852648.448\\
13	11087196.16\\
14	10745864.192\\
15	10586992.64\\
16	10487173.12\\
17	10448363.52\\
18	10426667.008\\
19	10416128\\
20	10416128\\
};
\addlegendentry{\large  $K:32$}

\addplot [color=mycolor3, line width=2.0pt, mark=triangle*, mark size=3.5000pt, mark options={solid, mycolor3}]
  table[row sep=crcr]{%
1	11271446.528\\
2	13270691.84\\
3	18593533.952\\
4	31245185.024\\
5	60656742.4\\
6	101071745.024\\
7	128441389.056\\
8	110924967.936\\
9	49416298.496\\
10	19918192.64\\
11	13694971.904\\
12	11852242.944\\
13	11087196.16\\
14	10745864.192\\
15	10586992.64\\
};

\addplot [color=mycolor3, line width=2.0pt, mark=triangle*, mark size=3.5000pt, mark options={solid, mycolor3}]
  table[row sep=crcr]{%
1	11271446.528\\
2	13270777.856\\
3	18695049.216\\
4	34661826.56\\
5	96432996.352\\
6	235181064.192\\
7	329938276.352\\
8	169393455.104\\
9	53307695.104\\
10	20007362.56\\
11	13697368.064\\
12	11852648.448\\
13	11087196.16\\
14	10745864.192\\
15	10586992.64\\
};

\addplot  [color=mycolor1, line width=2.0pt, dashed, mark=*, mark size=2.5000pt, mark options={solid, mycolor1}] table[row sep=crcr]{%
1	2051245.89568\\
2	2087084.15488\\
3	2174467.9936\\
4	2288286.76096\\
5	2446897.42848\\
6	2664899.51232\\
7	2970140.94848\\
8	3423407.872\\
9	4631535.19616\\
10	4066140.39552\\
11	3118653.44\\
12	2483097.6\\
13	2235985.92\\
14	2121064.448\\
15	2062118.912\\
16	2029953.024\\
17	2019635.2\\
18	2014945.28\\
19	2011136\\
20	2011136\\
};

\addplot [color=mycolor2, line width=2.0pt, dashed, mark=square*, mark size=2.5000pt, mark options={solid, mycolor2}]
  table[row sep=crcr]{%
1	4675492.92544\\
2	4900820.02944\\
3	5380669.31712\\
4	5999042.9696\\
5	6814914.94912\\
6	7865022.19776\\
7	9203625.30816\\
8	10978615.58272\\
9	13786867.85536\\
10	9311222.1696\\
11	6802042.88\\
12	5306892.288\\
13	4848414.72\\
14	4705730.56\\
15	4632113.152\\
16	4587171.84\\
17	4564615.168\\
18	4550656\\
19	4550656\\
20	4550656\\
};

\addplot [color=mycolor3, line width=2.0pt, dashed, mark=triangle*, mark size=3.5000pt, mark options={solid, mycolor3}]
  table[row sep=crcr]{%
1	11235613.4912\\
2	13219881.12384\\
3	16819175.99744\\
4	21318857.40032\\
5	26169306.9312\\
6	30711849.28768\\
7	37698764.8\\
8	40886484.992\\
9	33274740.736\\
10	18607386.624\\
11	13584486.4\\
12	11834064.896\\
13	11084025.856\\
14	10745864.192\\
15	10586992.64\\
16	10487173.12\\
17	10448363.52\\
18	10426667.008\\
19	10416128\\
20	10416128\\
};

\node(S_origin) at (axis cs:7.9, 4e7){};
\node(S_destination) at (axis cs:12,	4e7){};
\node(G_origin) at (axis cs:7.9, 1e9){};
\node(G_destination) at (axis cs:11,	1e9){};
\node(2_origin) at (axis cs:6.9, 1.2e8){};
\node(2_destination) at (axis cs:10.2,	1.2e8){};
\node(3_origin) at (axis cs:6.9, 3e8){};
\node(3_destination) at (axis cs:10.2,	3e8){};

\draw[-stealth,  line width=0.8mm, draw opacity=0.751] (S_origin)--(S_destination);
\draw[-stealth,  line width=0.8mm, draw opacity=0.751] (G_origin)--(G_destination);
\draw[-stealth,  line width=0.8mm, draw opacity=0.751] (2_origin)--(2_destination);
\draw[-stealth,  line width=0.8mm, draw opacity=0.751] (3_origin)--(3_destination);

\node[black] (SSA) at (axis cs:12.5, 4e7){\Large SSA};
\node[black](GSA) at (axis cs:11.5,	1e9){\Large GSA};
\node[black] (2L) at (axis cs:12, 1.2e8){\Large LGSA, $v_{\text{max}}: 2$};
\node[black](3L) at (axis cs:12,	3e8){\Large LGSA, $v_{\text{max}}: 4$};
\end{axis}
\end{tikzpicture}%

%% file: Sections/08.tex
\section{Conclusion and Outlook}\label{sec:discussion}
\acresetall 

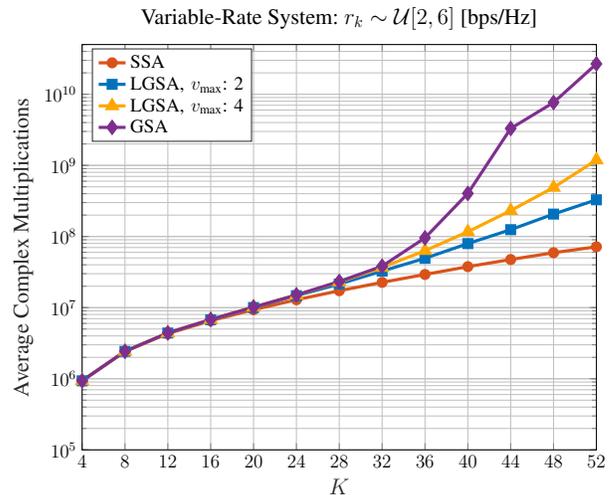
\begin{figure}[t]
    \centering
    \resizebox{0.45\textwidth}{!}{\input{Plots/VR_CompCount}}
    \caption{The average complex multiplications required to compute outage probability, $P_{\text{out}}$, for a variable-rate system, where $r_k$ is uniformly distributed between \qty{2}{} and \qty{6}{bps/Hz}. The results are obtained for a \ac{rec} with $M=64$ and $q_{\text{max}}=2$.}
    \label{fig:VR_compCount}
\end{figure}

In this paper, 
the information outage probability of the quasi-static \ac{noma}  channel under partial decoding is studied, where some aircraft can be decoded successfully while others experience outage. Our aim is minimizing the outage probability of the individual users (aircraft).
To this end, we propose three algorithms 
for \ac{sic} decoders and joint group decoders. 

The first algorithm is the \ac{ssa}, which identifies the optimal decoding order that minimizes the outage probability for \ac{sic} decoders. We demonstrate that \ac{ssa} outperforms the \ac{sic} decoding order proposed in \cite{order2}, and, when considering a variable-rate system, \ac{ssa} also surpasses the solution proposed in \cite{vblast}.
The second algorithm is \ac{gsa}, which uses \ac{sic} decoders together with joint group decoders to find the minimum achievable outage probability for indidividual aircraft. While \ac{gsa} significantly improves spectral efficiency compared to \ac{ssa}, a key concern is its feasibility in real-world applications, as joint decoding of a large group of aircraft might not be viable at the time being. To address this concern, we propose a suboptimal solution: the \ac{lgsa}, which restricts the number of aircraft to be decoded jointly. 
Even when the largest group of aircraft is limited to two, \ac{lgsa} significantly reduces the outage probability compared to the  \ac{sic}-only decoders such as \ac{ssa}.

As a result, \ac{lgsa} shows great promise.  In our future work, we will combine the receiver strategies for multiuser joint decoding proposed in \cite{joint1, joint2} with the decoding order obtained through \ac{lgsa}.

%% file: Plots/VR_CompCount.tex
%
%
\definecolor{mycolor1}{rgb}{0.85098,0.32549,0.09804}
\definecolor{mycolor2}{rgb}{0.00000,0.44706,0.74118}
\definecolor{mycolor3}{rgb}{1.0, 0.65, 0.0}
\definecolor{mycolor4}{rgb}{0.49412,0.18431,0.55686}
\definecolor{mycolor5}{rgb}{0.46667,0.67451,0.18824}
\pgfplotsset{compat=1.10}
\begin{tikzpicture}

\begin{axis}[%
width=4.521in,
height=3.566in,
at={(0.758in,0.481in)},
scale only axis,
title  ={\Large Variable-Rate System: $r_k\sim\mathcal{U}[2, 6]$ [bps/Hz]},
xmin=4,
xmax=52,
xtick={4,8,12,...,64},
xticklabel style = {font=\large},
xlabel style={font=\color{white!15!black}},
xlabel={\Large $K$},
ymode=log,
yminorticks=true,
ymin=1e5,
ymax=5e10,
ylabel style={font=\color{white!15!black}},
yticklabel style = {font=\large},
ylabel={\Large Average Complex Multiplications},
axis background/.style={fill=white},
xmajorgrids,
yminorgrids,
ymajorgrids,
legend style={at={(0.33,0.99)}, legend cell align=left, align=left, draw=white!15!black}
]

\addplot [color=mycolor1, line width=2.0pt, mark=*, mark size=2.5000pt, mark options={solid, mycolor1}]
  table[row sep=crcr]{%
4	934604.8\\
8	2421186.56\\
12	4311830.528\\
16	6531940.352\\
20	9359454.208\\
24	12886327.296\\
28	17275047.936\\
32	22657073.152\\
36	29262897.152\\
40	37672534.016\\
44	47479009.28\\
48	59166306.304\\
52	71654203.392\\
56	85551640.576\\
60	0\\
64	0\\
};
\addlegendentry{\large SSA}

\addplot [color=mycolor2, line width=2.0pt, mark=square*, mark size=2.5000pt, mark options={solid, mycolor2}]
  table[row sep=crcr]{%
4	935911.424\\
8	2429177.856\\
12	4395048.96\\
16	6755450.88\\
20	10033008.64\\
24	14729748.48\\
28	21463896.064\\
32	32623562.752\\
36	49406644.224\\
40	79484854.272\\
44	124818087.936\\
48	206588088.32\\
52	328287420.416\\
56	 0.5156e9\\
60	0.7891e9\\
64	1.0860e9\\
};
\addlegendentry{\large LGSA, $v_{\text{max}}\text{: 2}$}

\addplot [color=mycolor3, line width=2.0pt, mark=triangle*, mark size=3.5000pt, mark options={solid, mycolor3}]
  table[row sep=crcr]{%
4	935911.424\\
8	2429190.144\\
12	4408401.92\\
16	6805446.656\\
20	10153152.512\\
24	15043354.624\\
28	22901776.384\\
32	36673716.224\\
36	63242645.504\\
40	116417998.848\\
44	231278948.352\\
48	489203363.84\\
52	1189409714.176\\
56	0.2867e10\\
60	0.6450e10\\
64	1.6427e10\\
};
\addlegendentry{\large LGSA, $v_{\text{max}}\text{: 4}$}

\addplot [color=mycolor4, line width=2.0pt, mark=diamond*, mark size=3.5000pt, mark options={solid, mycolor4}]
  table[row sep=crcr]{%
4	935911.424\\
8	2429190.144\\
12	4408401.92\\
16	6805467.136\\
20	10158829.568\\
24	15072776.192\\
28	23298514.944\\
32	38279356.416\\
36	95697424.384\\
40	403006050.304\\
44	3301087170.56\\
48	7654100545.536\\
52	26732584894.464\\
56	0\\
60	0\\
64	0\\
};
\addlegendentry{\large GSA}

\end{axis}
\end{tikzpicture}%

%% file: Sections/Acknowledgement.tex
\section*{ACKNOWLEDGMENT}
The authors acknowledge Alexander Steingass for his important contributions to the conceptualization and initial development of this work. Although he passed away prior to submission, his scientific insight continues to influence this study. We remember him with respect and appreciation.


%% file: Sections/Appendix.tex
\appendices
 \section{\rewtwo{Impact of the \textit{PruneSubsets} Phase on the Computational Complexity of Algorithm \ref{al:step2}}}\label{ap:q_max}
 \rewtwo{In this section, we analyze how the \textit{PruneSubsets} phase impacts the computational complexity of Algorithm~\ref{al:step2} (\ac{gsa}) and justify our choice of setting $q_{\text{max}}=2$. To this end, we plot the average number of complex multiplications for different values of $q_{\text{max}}$. We also include the case in which the \textit{PruneSubsets} phase is omitted; that is, \ac{gsa} proceeds directly to the \textit{GreedyGroup} phase after \ac{ssa}. In the plots, this configuration is labeled as “No \textit{PruneSubset}”.}
 
 \rewtwo{It is important to emphasize that, for all values of $q_{\text{max}}$, including the case where the \textit{PruneSubsets} phase is omitted, the resulting outage probability remains identical. These configurations only influence the computational complexity of Algorithm~\ref{al:step2}; they do not affect the optimality of the final solution, since optimality is guaranteed in the \textit{GreedyGroup} phase (see Appendix~\ref{ap:gsa}). }
 
 \rewtwo{ In Fig.~\ref{fig:R2_C2_K16}, we plot the computational complexity for an equal-rate system with $K=16$ and $M=64$. The results show that the computational complexity of the algorithm remains relatively stable across all values of $q_{\text{max}}$. In contrast, when the \textit{PruneSubsets} phase is omitted, the computational complexity increases noticeably for $r_G > 6$~bps/Hz. Therefore, we conclude that the \textit{PruneSubsets} phase effectively reduces computational complexity of Algorithm~\ref{al:step2} in this case.}
 

\begin{figure}
     \centering
     \resizebox{0.4\textwidth}{!}{\input{Plots/appendix/K16_M64}}
     \caption{\rewtwo{Impact of the \textit{PruneSubsets} Phase on the Computational Complexity of Algorithm \ref{al:step2} for an equal-rate system with $K=16$, $M=64$}}
     \label{fig:R2_C2_K16}
\end{figure}

\begin{figure}
     \centering
     \resizebox{0.4\textwidth}{!}{\input{Plots/appendix/K32_M64}}
     \caption{\rewtwo{Impact of the \textit{PruneSubsets} Phase on the Computational Complexity of Algorithm \ref{al:step2} for an equal-rate system with $K=32$, $M=64$}}
     \label{fig:R2_C2_K32}
\end{figure}

 \rewtwo{Figure~\ref{fig:R2_C2_K32} presents the computational complexity  for an equal-rate system with $K=32$ and $M=64$. In this case, the value of $q_{\text{max}}$ that yields the lowest computational complexity depends on $r_G$.
For $r_G \leq 5$, we observe that lower values of $q_{\text{max}}$ lead to lower computational complexity. As a result, the best option is to omit the \textit{PruneSubset} phase entirely. The second-best option is $q_{max}=2$, which yields a very similar computational complexity, with only a slight increase compared to omitting the \textit{PruneSubset} phase.
On the other hand, for $r_G \geq 6$, the algorithm without the \textit{PruneSubset} phase becomes computationally very demanding, while with $q_{max}=2$, the computational complexity of the system is already significantly reduced. As $q_{\text{max}}$ increases from 2 to 5, the computational complexity continues to decrease. However, as $q_{\text{max}}$ increases further from 5 to 10, the computational complexity begins to rise again. Consequently, for $r_G \leq 5$, the optimal approach is to omit the \textit{PruneSubset} phase, while for $r_G \geq 6$, the optimal choice is $q_{max}=5$.}

\rewtwo{Figures~\ref{fig:R2_C2_K16} and~\ref{fig:R2_C2_K32} demonstrate that the value of $q_{\text{max}}$ that minimizes the computational complexity in Algorithm~\ref{al:step2} depends on the parameters $K$, $M$, and the transmission rates of the aircraft, $r_k$.
For an equal-rate system with $K=32$ and $M=64$, we highlight that for $r_G \leq 5$~bps/Hz, the outage probability remains below 0.06. For $r_G\geq 6$~bps/Hz, the outage probability increases significantly. 
Since practical communication systems are typically designed to operate with low outage probabilities (e.g., below 0.05), it is reasonable to choose the value of $q_{\text{max}}$ that minimizes the computational complexity within this operating range.
At the same time, it is important to select a $q_{\text{max}}$ that still permits  evaluation of outage probabilities at higher transmission rates. Setting $q_{\text{max}}=2$ provides an effective compromise: it significantly reduces computational complexity at transmission rates that result in low outage probabilities, while still enabling performance analysis across a broader range of transmission rates.}
 

 \section{Proof of the Optimality of Algorithm \ref{al:step2}}\label{ap:gsa}
Algorithm \ref{al:step2}, \ac{gsa}, solves the problem given in  \eqref{eq:probGSA} by finding the maximal set $S^*$ that satisfies the condition given in~\eqref{eq:con_GSA}.
\rewone{The algorithm produces three sets: $S^*$, $\hat{S}$, and $L$, where $S^*$ contains the aircraft that can be successfully decoded, and $\hat{S}$ includes the aircraft that are definitely in outage, as justified in the description of the algorithm in Section~\ref{sec:algs}.
To demonstrate the optimality of Algorithm~\ref{al:step2}, it is sufficient to prove that all aircraft in $L$ are in outage, regardless of the decoding order or the joint group decoding strategy.}

\rewone{The algorithm terminates in two cases: 1) When \mbox{$|L|=0$}, where set $L$ is empty. 
2) When \mbox{$|L|<v$}, where $v$ is the group-size in joint group decoding. Since it is not possible to form a group of users from $L$ with a size that exceeds $|L|$, once $v$ becomes larger than $|L|$, the algorithm terminates.
In the case when  \mbox{$|L|=0$}, it follows trivially that $S^*$ is the optimal set. Therefore, we focus on $1\leq|L|$.}
To this end, before termination, the algorithm ensures that none of the combinations of aircraft in $L$ satisfy the condition given in line~\ref{gsa:line:constraint} of Algorithm~\ref{al:step2}. This means 
\rewone{\begin{equation}\label{eq:appB1}
    \not\exists\, C \subseteq L \quad \text{such that} \quad 
\forall\, S \subseteq C,\ 
\sum_{k\in S} r_{k} \le R_{C,S}^{T_C} \text{ .}
\end{equation}}
\rewone{In what follows, we prove that when the condition in \eqref{eq:appB1} holds, then all the aircraft in $L$ are definitely in outage.}
\begin{lemma}
\rewone{
Let sets $S^*$, $\hat{S}$, and $L$ form a partition of $\{1, 2,\dots,K\}$. The set $S^*$ contains the aircraft that are successfully decoded and the set $\hat{S}$ denotes the outage set, i.e., the aircraft in $\hat{S}$ are certainly in outage. 
Let $C\subseteq L$ and let \mbox{$T_C=(L\cup\hat{S})\symbol{92}\{C\}$} denote the set containing all aircraft in $L$ and $\hat{S}$ except $C$. Moreover,  define $R_{C,S}^{T_C}$ as the achievable sum rate of aircraft in $S$, where $S\subseteq C$, under the interference of the aircraft in $T_C$. Given these definitions, if \eqref{eq:appB1} holds, then all aircraft in $L$ are in outage. }
\end{lemma}
\begin{proof}
\rewone{The statement can be proved by contradiction. Suppose that there exists a subset of aircraft $C \subseteq L$ that can be decoded while treating the aircraft in $T_C$ as interference. Recall that $R^{T_C}_{C,S}$ is the mutual information between the transmit signals and the received signal of the set $S$, conditioned on the signals in \mbox{$T_C=(L\cup\hat{S})\symbol{92}\{C\}$}. Decodability of the set $C$ under these conditions implies that the rate point $\{r_k : k \in C\}$ is inside the MAC capacity region of the channel where all signals coming from aircraft in $S^*$ have been removed and all signals coming from aircraft in $T_C$ are treated as interference. Hence, the set $C$
must satisfy the following condition
\begin{equation}\label{eq:appB2}
\forall\, S \subseteq C,\ 
\sum_{k\in S} r_{k} \le R_{C,S}^{T_C} \text{ .}
\end{equation}
However, this is not possible since by \eqref{eq:appB1} such a subset \mbox{$C \subseteq L$} does not exist.
This means that there exists no subsets $C$ of $L$ such that all users in $C$ can be decoded by any possible decoding scheme.
This proves the statement.}
\end{proof}

%% file: Plots/appendix/K16_M64.tex
%
%
\definecolor{mycolor4}{rgb}{0.49412,0.18431,0.55686}
\definecolor{mycolor2}{rgb}{0.85098,0.32549,0.09804}
\definecolor{mycolor3}{rgb}{0.00000,0.44706,0.74118}
\definecolor{mycolor1}{rgb}{1.0, 0.65, 0.0}
\definecolor{mycolor9}{rgb}{0.30196,0.74510,0.93333}
\definecolor{mycolor6}{rgb}{0.63529,0.07843,0.18431}
\definecolor{mycolor7}{rgb}{0.46667,0.67451,0.18824}
\definecolor{mycolor8}{rgb}{0.9, 0.17, 0.31} 
\definecolor{mycolor5}{rgb}{0.0, 0.0, 1.0} 

\pgfplotsset{compat=1.10}
\begin{tikzpicture}

\begin{axis}[%
width=4.521in,
height=3.566in,
at={(0.758in,0.481in)},
scale only axis,
title  ={\LARGE Equal-Rate System: $r_k=r_G, k\in [K]$},
xmin=1,
xmax=10,
xtick={1, 2, 3, ..., 15},
xticklabel style = {font=\Large},
xlabel style={font=\color{white!15!black}},
xlabel={\LARGE $r_G$ [bps/Hz]},
ymode=log,
yminorticks=true,
ymin=4e6,
ymax=9e7,
ylabel style={font=\color{white!15!black}},
yticklabel style = {font=\Large},
ylabel={\LARGE Average Complex Multiplications},
axis background/.style={fill=white},
xmajorgrids,
yminorgrids,
ymajorgrids,
legend style={at={(0.39,0.99)}, legend cell align=left, align=left, draw=white!15!black}
]

\addplot [color=mycolor1, line width=2.5pt,  mark options={solid, mycolor1}] table[row sep=crcr]{%
1	4673118.208\\
2	4913758.208\\
3	5470883.84\\
4	6296952.832\\
5	7900258.304\\
6	11830386.688\\
7	21210382.336\\
8	53404442.624\\
9	59176394.752\\
10	11584528.384\\
};
\addlegendentry{\Large No \textit{PruneSubset}}

\addplot [color=mycolor2, line width=2.5pt, mark options={solid, mycolor2}]
  table[row sep=crcr]{%
1	4673118.208\\
2	4914769.92\\
3	5481271.296\\
4	6324240.384\\
5	8007323.648\\
6	10993844.224\\
7	15745318.912\\
8	19478777.856\\
9	18704314.368\\
10	10521300.992\\
};
\addlegendentry{\Large $q_{max}: 2$}

\addplot [color=mycolor3, line width=2.5pt,  mark options={solid, mycolor3}]
  table[row sep=crcr]{%
1	4673118.208\\
2	4914905.088\\
3	5486632.96\\
4	6342217.728\\
5	8103034.88\\
6	11369914.368\\
7	15527231.488\\
8	18278215.68\\
9	18644672.512\\
10	10458886.144\\
};
\addlegendentry{\Large  $q_{max}: 3$}

\addplot [color=mycolor4, line width=2.5pt,  mark options={solid, mycolor4}]
  table[row sep=crcr]{%
1	4673118.208\\
2	4914987.008\\
3	5488771.072\\
4	6353604.608\\
5	8170340.352\\
6	11561562.112\\
7	15496282.112\\
8	18454781.952\\
9	18823168\\
10	10477547.52\\
};
\addlegendentry{\Large  $q_{max}: 4$}

\addplot [color=mycolor5, line width=2.5pt,  mark options={solid, mycolor5}]
  table[row sep=crcr]{%
1	4673118.208\\
2	4915007.488\\
3	5489508.352\\
4	6359441.408\\
5	8180502.528\\
6	11894370.304\\
7	15942062.08\\
8	18534305.792\\
9	18954928.128\\
10	10483228.672\\
};
\addlegendentry{\Large  $q_{max}: 5$}

\addplot  [color=mycolor6, line width=2.5pt, mark options={solid, mycolor6}] table[row sep=crcr]{%
1	4673118.208\\
2	4915007.488\\
3	5489729.536\\
4	6361751.552\\
5	8204529.664\\
6	11992948.736\\
7	16273891.328\\
8	18629742.592\\
9	19035967.488\\
10	10483908.608\\
};
\addlegendentry{\Large  $q_{max}: 6$}

\addplot [color=mycolor8, line width=2.5pt,  mark options={solid, mycolor8}]
  table[row sep=crcr]{%
1	4673118.208\\
2	4915007.488\\
3	5489758.208\\
4	6362390.528\\
5	8180772.864\\
6	12235767.808\\
7	16436158.464\\
8	18676248.576\\
9	19073187.84\\
10	10483908.608\\
};
\addlegendentry{\Large  $q_{max}: 8$}

\addplot [color=mycolor9, line width=2.5pt, mark options={solid, mycolor9}]
  table[row sep=crcr]{%
1	4673118.208\\
2	4915007.488\\
3	5489758.208\\
4	6362390.528\\
5	8181592.064\\
6	12270309.376\\
7	16465489.92\\
8	18677014.528\\
9	19073896.448\\
10	10483908.608\\
};
\addlegendentry{\Large  $q_{max}: 10$}

\end{axis}
\end{tikzpicture}%

%% file: Plots/appendix/K32_M64.tex
%
%
\definecolor{mycolor4}{rgb}{0.49412,0.18431,0.55686}
\definecolor{mycolor2}{rgb}{0.85098,0.32549,0.09804}
\definecolor{mycolor3}{rgb}{0.00000,0.44706,0.74118}
\definecolor{mycolor1}{rgb}{1.0, 0.65, 0.0}
\definecolor{mycolor9}{rgb}{0.30196,0.74510,0.93333}
\definecolor{mycolor6}{rgb}{0.63529,0.07843,0.18431}
\definecolor{mycolor7}{rgb}{0.46667,0.67451,0.18824}
\definecolor{mycolor8}{rgb}{0.9, 0.17, 0.31} 
\definecolor{mycolor5}{rgb}{0.0, 0.0, 1.0} 

\pgfplotsset{compat=1.10}
\begin{tikzpicture}

\begin{axis}[%
width=4.521in,
height=3.566in,
at={(0.758in,0.481in)},
scale only axis,
title  ={\LARGE Equal-Rate System: $r_k=r_G, k\in [K]$},
xmin=1,
xmax=10,
xtick={1, 2, 3, ..., 15},
xticklabel style = {font=\Large},
xlabel style={font=\color{white!15!black}},
xlabel={\LARGE $r_G$ [bps/Hz]},
ymode=log,
yminorticks=true,
ymin=1e7,
ymax=1e12,
ylabel style={font=\color{white!15!black}},
yticklabel style = {font=\Large},
ylabel={\LARGE Average Complex Multiplications},
axis background/.style={fill=white},
xmajorgrids,
yminorgrids,
ymajorgrids,
legend style={at={(0.38,0.99)}, legend cell align=left, align=left, draw=white!15!black}
]

\addplot [color=mycolor1, line width=2.5pt,  mark options={solid, mycolor1}] table[row sep=crcr]{%
1	11271446.528\\
2	13298610.176\\
3	18839572.48\\
4	35556139.008\\
5	181687128.064\\
6	9715730583.552\\
7	311666400501.76\\
8	828770244280.32\\
9	323935248015.36\\
10	26253516.8\\
};
\addlegendentry{\Large No \textit{PruneSubset}}

\addplot [color=mycolor2, line width=2.5pt, mark options={solid, mycolor2}]
  table[row sep=crcr]{%
1	11271446.528\\
2	13314154.496\\
3	19044020.224\\
4	37145612.288\\
5	181242707.968\\
6	2179319341.056\\
7	12442417180.672\\
8	2238738579.456\\
9	64222412.8\\
10	20100501.504\\
};
\addlegendentry{\Large $q_{max}: 2$}

\addplot [color=mycolor3, line width=2.5pt,  mark options={solid, mycolor3}]
  table[row sep=crcr]{%
1	11271446.528\\
2	13336813.568\\
3	19494227.968\\
4	41993867.264\\
5	181033660.416\\
6	839365189.632\\
7	4597903859.712\\
8	440652599.296\\
9	55879692.288\\
10	20069044.224\\
};
\addlegendentry{\Large  $q_{max}: 3$}

\addplot [color=mycolor4, line width=2.5pt,  mark options={solid, mycolor4}]
  table[row sep=crcr]{%
1	11271446.528\\
2	13361438.72\\
3	20348329.984\\
4	53434654.72\\
5	266158202.88\\
6	770791383.04\\
7	1948346413.056\\
8	345952866.304\\
9	53553197.056\\
10	20092891.136\\
};
\addlegendentry{\Large  $q_{max}: 4$}

\addplot [color=mycolor5, line width=2.5pt,  mark options={solid, mycolor5}]
  table[row sep=crcr]{%
1	11271446.528\\
2	13381386.24\\
3	21735161.856\\
4	76506312.704\\
5	493485670.4\\
6	1197445828.608\\
7	808443248.64\\
8	342074769.408\\
9	53656780.8\\
10	20103163.904\\
};
\addlegendentry{\Large  $q_{max}: 5$}

\addplot  [color=mycolor6, line width=2.5pt, mark options={solid, mycolor6}] table[row sep=crcr]{%
1	11271446.528\\
2	13393965.056\\
3	23638249.472\\
4	115925925.888\\
5	928183209.984\\
6	2212717670.4\\
7	1034972516.352\\
8	291355799.552\\
9	54064795.648\\
10	20105973.76\\
};
\addlegendentry{\Large  $q_{max}: 6$}

\addplot [color=mycolor8, line width=2.5pt,  mark options={solid, mycolor8}]
  table[row sep=crcr]{%
1	11271446.528\\
2	13402148.864\\
3	27898966.016\\
4	248516001.792\\
5	2828811227.136\\
6	7602858037.248\\
7	2399208079.36\\
8	629899448.32\\
9	54660427.776\\
10	20058865.664\\
};
\addlegendentry{\Large  $q_{max}: 8$}

\addplot [color=mycolor9, line width=2.5pt, mark options={solid, mycolor9}]
  table[row sep=crcr]{%
1	11271446.528\\
2	13402632.192\\
3	30567608.32\\
4	400583143.424\\
5	6047633747.968\\
6	19181230354.432\\
7	4724808986.624\\
8	1455823474.688\\
9	54781546.496\\
10	20058865.664\\
};
\addlegendentry{\Large  $q_{max}: 10$}

\end{axis}
\end{tikzpicture}%